\documentclass[11pt,a4paper]{amsart}

\usepackage{hyperref}
\usepackage{cite}
\usepackage{varioref}
\usepackage{graphicx}
\usepackage[all,cmtip]{xy}
\usepackage{amssymb}

\hypersetup{colorlinks=true, citecolor=blue}

\newtheorem{theorem}{Theorem}[section]
\newtheorem{lemma}[theorem]{Lemma}

\theoremstyle{definition}
\newtheorem{definition}[theorem]{Definition}
\newtheorem{example}[theorem]{Example}

\theoremstyle{remark}
\newtheorem{remark}[theorem]{Remark}

\numberwithin{equation}{section}

\def\bbR{\mathbb{R}}
\def\bbS{\mathbb{S}}
\def\bbT{\mathbb{T}}
\def\bw{\begingroup\textstyle\bigwedge\endgroup}

\def\lra{\longrightarrow}
\def\la{\langle}
\def\ra{\rangle}
\def\tint{\begingroup\textstyle\int\endgroup}

\def\mfo{\mathfrak{o}}
\def\mft{\mathfrak{t}}
\def\vol{\mathrm{vol}}
\def\id{\mathrm{id}}

\def\supp{\mathrm{supp}}
\def\ker{\mathrm{ker}}
\def\dd{\mathrm{d}}
\def\de{\delta}

\def\S{\mathcal{S}}
\def\SscA{\S_{sc\,A}}
\def\SscF{\S_{sc\,F}}
\def\G{\mathcal{G}}
\def\GscA{\G_{sc\,A}}
\def\O{\mathcal{O}}
\def\E{\mathcal{E}}
\def\Ekin{\E^\mathrm{kin}}
\def\Einv{\E^\mathrm{inv}}

\def\c{\mathrm{C}^\infty}
\def\cc{\c_\mathrm{c}}
\def\csc{\c_\mathrm{sc}}
\def\ctc{\c_\mathrm{tc}}
\def\f{\Omega}
\def\fdd{\f_\dd}
\def\fde{\f_\de}
\def\fc{\f_\mathrm{c}}
\def\fcdd{\f_{\mathrm{c}\,\dd}}
\def\fcde{\f_{\mathrm{c}\,\de}}
\def\fsc{\f_\mathrm{sc}}
\def\fscdd{\f_{\mathrm{sc}\,\dd}}
\def\fscde{\f_{\mathrm{sc}\,\de}}
\def\fpc{\f_\mathrm{pc}}
\def\ffc{\f_\mathrm{fc}}
\def\ftc{\f_\mathrm{tc}}
\def\ftcdd{\f_{\mathrm{tc}\,\dd}}
\def\ftcde{\f_{\mathrm{tc}\,\de}}
\def\hdd{\mathrm{H}_\dd}
\def\hde{\mathrm{H}_\de}
\def\hcdd{\mathrm{H}_{\mathrm{c}\,\dd}}
\def\hcde{\mathrm{H}_{\mathrm{c}\,\de}}
\def\hscdd{\mathrm{H}_{\mathrm{sc}\,\dd}}
\def\hscde{\mathrm{H}_{\mathrm{sc}\,\de}}
\def\htcdd{\mathrm{H}_{\mathrm{tc}\,\dd}}
\def\htcde{\mathrm{H}_{\mathrm{tc}\,\de}}

\newcommand{\quotes}[1]{``#1''}

\begin{document}

%\title[Cohomologies with restricted support and Maxwell $k$-forms]
%{De Rham cohomologies with restricted support and observables for Maxwell $k$-forms}
\title[Optimal space of classical observables for Maxwell $k$-forms]
{Optimal space of linear classical observables for Maxwell $k$-forms 
via spacelike and timelike compact de Rham cohomologies}

%    author information
\author{Marco Benini}
\address{Dipartimento di Fisica, Universit\`a di Pavia 
\&~INFN, Sezione di Pavia -- Via Bassi~6, I-27100 Pavia, Italy}
\curraddr{}
\email{marco.benini@pv.infn.it}
\thanks{}



\date{\today}

\begin{abstract}
Being motivated by open questions in gauge field theories, we consider non-standard de Rham cohomology groups 
for timelike compact and spacelike compact support systems. These cohomology groups are shown 
to be isomorphic respectively to the usual de Rham cohomology of a spacelike Cauchy surface 
and its counterpart with compact support. Furthermore, an analog of the usual Poincar\'e duality 
for de Rham cohomology is shown to hold for the case with non-standard supports as well. 
We apply these results to find {\em optimal} spaces of linear observables 
for analogs of arbitrary degree $k$ of both the vector potential and the Faraday tensor. 
%the classical field theoretical models 
%$\de\dd A=0$ and $\dd F=0,\,\de F=0$, both $A$ and $F$ being $k$-forms for $k\in\{1,\dots,m-1\}$ 
%over a $m$-dimensional globally hyperbolic spacetime $M$ and regarding $A$ and $A^\prime$ 
%as equivalent provided $A^\prime-A\in\dd\f^{k-1}(M)$. 
The term {\em optimal} has to be intended in the following sense: 
The spaces of linear observables we consider distinguish between different configurations; 
in addition to that, there are no redundant observables. %This result is achieved via cohomological arguments. 
This last point in particular heavily relies on the %is shown introducing non-standard de Rham cohomology groups 
%with spacelike compact and timelike compact support and analyzing their properties. 
%Specifically, these cohomology groups are shown to be isomorphic 
%to ordinary de Rham cohomologies of a spacelike Cauchy surface. 
%Furthermore, an 
analog of Poincar\'e duality for the new cohomology groups.% is shown to hold. 
\end{abstract}

\maketitle

\section{Introduction}
Recently, there have been several attempts to deal with the algebraic quantization of electromagnetism 
and, more generally, of gauge theories in the framework of general local covariance \cite{BFV03}. 
It is worth to mention the attempt of \cite{Bon77,DL12} to quantize electromagnetism via the Faraday tensor 
and the various attempts to deal with the gauge theory of the vector potential 
\cite{Dim92,FP03,Pfe09,Dap11,DS13,SDH12,FS13}. 
Furthermore, general linear gauge theories have been analyzed in \cite{HS13}, while a thorough analysis 
of electromagnetism as a $U(1)$ Yang-Mills model can be found in \cite{BDS13b,BDHS13}. 
Many details about general local covariance in the broader family 
of quasilinear field theories are discussed in \cite{Kha12}. 

At some point, all these different approaches exhibit a certain sensitivity 
to specific cohomological properties of the background spacetime causing general local covariance not to hold. 
In particular, several slightly different approaches in defining the space of classical observables 
have been considered in order to recover general local covariance. Already at a classical level, 
this fact raises the question of finding an optimal space of observables for a given field theory. 

Up to now modifications to the space of classical observables have been considered 
with the aim of recovering general local covariance. More specifically, mild variations in the classical models 
have been taken into account in the attempt to prevent spacetime embeddings 
to give rise to non-injective morphisms already at the level of the classical observables. 
Here we would like to assume a different, more intrinsic, approach to the problem of the choice 
of the most suitable space of functionals to be regarded as observables for classical field configurations. 

According to \cite{BFR12}, one regards classical observables as functionals on the space of field configurations. 
In particular, when dealing with linear equations of motion, 
as a starting point one considers linear functionals defined on off-shell field configurations. 
In case the model exhibits a gauge symmetry, only gauge invariant functionals are taken into account. 
At this point one implements also the dynamics in a dual fashion, 
thus the evaluation of the resulting linear functionals makes sense only for (gauge classes of) on-shell configurations. 
The question is now the following: Is the space of functionals obtained according to these prescriptions 
optimal for the classical field theoretical model under analysis? 
We adopt the criteria stated below in order to give a precise meaning 
to the term {\em optimal} employed in the last sentence: 
\begin{description}
\item[Separability of configurations] One might ask whether the resulting space of functionals 
is sufficiently rich in order to distinguish all possible classical field configurations. 
As a matter of fact, it would be unsatisfactory, at least from a classical perspective, 
not to be able to distinguish solutions to the equation of motion 
which are not the same (or which are not regarded as being equivalent due to gauge symmetry). 
\item[Non-redundancy of observables] Another relevant point is redundancy of the resulting space of functionals 
on (gauge classes of) on-shell field configurations. This amounts to require that 
no elements in this space provide the same outcomes upon evaluation 
on all possible solutions to the field equations (eventually, up to gauge). 
If this is not the case, then one is over-counting observables. 
\end{description}
Both requirements trivially hold true for any field theory ruled by a Green-hyperbolic linear differential operator, 
see\ e.g.\ \cite{BG12a,BG12b} for a detailed description of these models also at a quantum level. 
The same conclusion can be easily extended to the affine field theories considered in \cite{BDS13a}. 
They tipically become non-trivial for systems ruled by non-hyperbolic equations with an underlying gauge symmetry. 
In this case separability of field configurations might be unclear due to the restriction to gauge invariant functionals 
one is forced to take into account in defining observables for a gauge theory. 
As for non-redundancy of observables, it is related to the fact that 
on-shell configurations for non-hyperbolic equations are obtained from initial data with certain restrictions. 
For example, separability of configurations was shown to fail for the space of gauge invariant affine functionals 
considered in \cite{BDS13b}, where electromagnetism is regarded as a $U(1)$ Yang-Mills model. 
This failure suggested to modify the structure of the functionals in order to recover the ability 
to separate gauge classes of on-shell field configurations, thus leading to the results of \cite{BDHS13}. 
To the best of our knowledge, the requirement of non-redundancy for the space of observables 
was never explicitly considered before. 

The question we address in this paper is whether separability of field configurations 
and non-redundancy of the space of observables hold true for 
the spaces of linear functionals obtained along the lines sketched above 
when dealing with analogs of the Faraday tensor and of the vector potential in arbitrary degree. 
This would allow us to interpret these spaces as being optimal (with respect to the criteria stated above) 
to define classical observables for the Faraday tensor and for the vector potential in arbitrary degree. 
It turns out that separability of field configurations is trivial for the Faraday tensor, 
while it follows from Poincar\'e duality between de Rham cohomology 
and its counterpart with compact support in the case of the vector potential. 
As for non-redundancy of observables, the problem is more subtle. 
In particular, it is related to certain analogs of the de Rham cohomology groups 
defined out of forms with support contained inside spacelike compact and respectively timelike compact regions. 

We study these cohomology groups with non-standard support in detail in the first part of the paper. 
In particular, we exhibit isomorphisms between spacelike compact and timelike compact cohomologies and the 
usual de Rham cohomologies  of a spacelike Cauchy surface with respectively compact and non-compact support. 
Furthermore, a relation similar to the usual Poincar\'e duality for de Rham cohomologies is show to hold 
between the new cohomology groups. This result plays a central role in proving non-redundancy 
for the space of linear functionals we consider in the case of the Faraday tensor 
as well as the one taken into account when dealing with the gauge theory of the vector potential. 

Unfortunately cohomology groups with spacelike compact or timelike compact support 
do not exhibit a manifest functorial nature. This is due to the fact that 
forms with spacelike compact and timelike compact support do not behave well 
neither in a covariant nor in a contravariant way. 
In this respect, the author would like to mention an upcoming work \cite{Kha14}, where results 
strongly related to those of the present paper are achieved by means of different techniques connected 
to the dynamics of the Laplace-de Rham operator $\Box=\de\dd+\dd\de$ on a globally hyperbolic spacetime. 
This approach should clarify whether and to what extent spacelike compact and timelike compact cohomologies 
fail in showing a functorial behavior. 

The paper is organized as follows: In Section\ \ref{secLorGeom} we recall the definition of 
a globally hyperbolic spacetime, together with some features which are relevant to the present discussion. 
Section\ \ref{secFormsCohomologyDuality} is devoted to provide the mathematical background 
for the study of de Rham cohomology. We get to the point where sufficient structure is available in order to  
to formulate Poincar\'e duality between de Rham cohomology groups with compact support 
and de Rham cohomology groups with restriction on the support. 
In Section\ \ref{secSCCohomology}, spacelike compact cohomology groups 
are introduced and an isomorphism with de Rham cohomology groups with compact support 
of a spacelike Cauchy surface is exhibited. The same procedure is followed in Section\ \ref{secTCCohomology} 
for timelike compact cohomology groups, which are shown to be isomorphic to de Rham cohomology groups 
of a spacelike Cauchy surface. Section\ \ref{secSCTCDuality} is devoted to prove an analog of Poincar\'e duality 
between spacelike compact and timelike compact cohomologies. We conclude with Section\ \ref{secObservables} 
studying the dynamics of the analogs in arbitrary degree of the vector potential 
($\de\dd A=0$ where a field configuration $A$ is given up to exact forms) 
and of the Faraday tensor ($\dd F=0$ and $\de F=0$). 
In particular, we introduce suitable spaces of (gauge invariant, in the case of the vector potential) linear functionals 
modulo equations of motion. Exploiting standard Poincar\'e duality as well as its counterpart 
for spacelike compact and timelike compact cohomologies, it is shown that 
the resulting quotients provide optimal spaces of classical observables 
with respect to the criteria of separability of field configurations and of non-redundancy of observables.

\section{Lorentzian geometry}\label{secLorGeom}
In this section we recollect some well-known facts about Lorentzian geometry
with a particular focus on globally hyperbolic spacetimes. 
As general references we follow \cite{BEE96,BGP07,Wal12}. 

From now on we will assume 
all manifolds to be Hausdorff, second countable, smooth, boundaryless, orientable and of dimension $m\geq2$. 
Furthermore, maps between manifolds are implicitly taken to be smooth. 

\begin{definition}
A Lorentzian manifold $(M,g)$ is a manifold $M$ endowed with a Lorentzian metric $g$, 
whose signature is of type $-+\cdots+$. 
\end{definition}

Lorentzian manifolds provide the appropriate background 
to distinguish tangent vectors $v\in T_pM$ at a point $p\in M$ among three classes: 
\begin{enumerate}
\item[a.] $v$ is {\em timelike} if $g(v,v)<0$; 
\item[b.] $v$ is {\em lightlike} if $g(v,v)=0$; 
\item[c.] $v$ is {\em spacelike} if $g(v,v)>0$. 
\end{enumerate}
{\em Causal} tangent vectors are either timelike or lightlike. 
This distinction can be used to classify certain curves $\gamma:I\to M$ on $M$ 
according to the behavior of the vector $\dot{\gamma}$ tangent to the curve: 
$\gamma$ is timelike, lightlike or causal if $\dot{\gamma}(s)$ is such for each $s\in I$. 

\begin{definition}\label{defSpacetime}
A Lorentzian manifold $(M,g)$ is time-orientable if there exists a vector field $\mft$ which is timelike everywhere. 
Any such $\mft$ provides a {\em time-orientation} for $(M,g)$. 

A spacetime is a quadruple $(M,g,\mfo,\mft)$, 
where $(M,g)$ is a time-orientable Lorentzian manifold, $\mfo$ is a choice of orientation for $M$ 
and $\mft$ is a choice of time-orientation. 
\end{definition}

Fixing a time-orientation $\mft$ enables one 
to distinguish non zero causal tangent vectors $0\neq v\in T_pM$ in two classes: 
\begin{enumerate}
\item[a.] $v$ is {\em future-directed} if $g(\mft,v)<0$; 
\item[b.] $v$ is {\em past-directed} if $g(\mft,v)>0$. 
\end{enumerate}
Accordingly, one can also define causal curves to be past-/future-directed 
if the vector tangent to the curve is future-/past-directed everywhere along the curve. 
We are thus in position to define the {\em causal future/past} $J_M^\pm(S)$ of a subset $S\subseteq M$ 
as the set of points which can be reached by a future-/past-directed causal curve emanating from $S$. 
Similarly, one defines the {\em chronological future/past} $I_M^\pm(S)$ of $S$ 
taking into account only timelike curves. 
Furthermore, it is conventional to introduce $J_M(S)\doteq J_M^+(S)\cup J_M^-(S)$ 
and similarly $I_M(S)\doteq I_M^+(S)\cup I_M^-(S)$. 

\begin{definition}
Let $(M,g,\mfo,\mft)$ be a spacetime. A {\em Cauchy surface }$\Sigma$ for $(M,g,\mfo,\mft)$ is a subset of $M$ 
intersecting each inextensible timelike curve exactly once. 
$(M,g,\mfo,\mft)$ is {\em globally hyperbolic} if there exists a Cauchy surface. 
\end{definition}

Note that, without any further assumption, a Cauchy surface is a $(m-1)$-dimensional submanifold of $M$ 
only in the topological sense. 
However, any Cauchy surface considered in the following 
is implicitly supposed to be a smooth submanifold, unless otherwise stated. 

One of the most prominent results about globally hyperbolic spacetimes is the following theorem \cite{BS05,BS06}. 

\begin{theorem}\label{thmGlobHyp}
On a spacetime $(M,g,\mfo,\mft)$ the statements below are equivalent:
\begin{enumerate}
\item[a.] $(M,g,\mfo,\mft)$ is globally hyperbolic;
\item[b.] $(M,g,\mfo,\mft)$ is isometric to $\bbR\times\Sigma$ 
endowed with the metric $-\beta\dd t\otimes\dd t+h_t$, 
where $\Sigma$ is a spacelike Cauchy surface for $(M,g,\mfo,\mft)$, 
$t:\bbR\times\Sigma\to\bbR$ is the projection on the first factor, $\beta\in\c(\bbR\times\Sigma)$ is strictly positive, 
$\bbR\ni t\mapsto h_t$ provides a smooth $1$-parameter family of Riemannian metrics on $\Sigma$
and $\{t\}\times\Sigma$ is a spacelike Cauchy surface of $(M,g,\mfo,\mft)$ for each $t\in\bbR$. 
\end{enumerate}

Furthermore, any spacelike Cauchy surface $\Sigma$ for $(M,g,\mfo,\mft)$ induces a foliation of $(M,g,\mfo,\mft)$ 
of the type described in statement b. 
\end{theorem}

From now on, we will often refer to a globally hyperbolic spacetime $(M,g,\mfo,\mft)$ 
explicitly denoting only the underlying manifold $M$, the rest of the data being understood. 

We conclude the present section recalling some nomenclature 
for subsets $S$ of a globally hyperbolic spacetime $M$: 
\begin{enumerate}
\item[{\bf pc/fc:}] $S\subseteq M$ is {\em past compact} ({\em future compact}) 
if $S\cap J_M^-(K)$ (respectively $S\cap J_M^+(K)$) is compact for each $K\subseteq M$ compact; 
\item[{\bf tc:}] $S\subseteq M$ is {\em timelike compact} if it is both past-compact and future-compact; 
\item[{\bf sc:}] $S\subseteq M$ is {\em spacelike compact} 
if it is closed and there exists $K\subseteq M$ compact such that $S\subseteq J_M(K)$. 
\end{enumerate}

\section{Forms, de Rham cohomology and Poincar\'e duality}\label{secFormsCohomologyDuality}
This section is devoted to recall some standard results 
about de Rham cohomology for arbitrary differential forms and differential forms with compact support. 
The final aim is to provide the sufficient background to state a version of Poincar\'e duality 
between cohomology and cohomology with compact support. 
We take the chance to set some notation which will be used later on. 
Our references are the classic books \cite{deR84,BT82}. 

As it is customary, the space of differential forms $\f^k(M)$ of degree $k$ on a $m$-dimensional manifold $M$ 
is introduced as the space of sections of the $k$-th exterior power 
$\bw^k(T^\ast M)$ of the cotangent bundle $T^\ast M$. 
Therefore $\f^\ast(M)=\bigoplus_k\f^k(M)$ is naturally endowed 
with the structure of a graded algebra with respect to the {\em wedge product} $\wedge$. 
One might also consider forms with compact support on $M$, denoted by $\fc^\ast(M)$. 

If $M$ is oriented, a natural notion of integral is defined for forms of top degree $k=m$. 
In particular, we have a bilinear pairing $\la\cdot,\cdot\ra$ between $k$-forms and $(m-k)$-forms: 
\begin{equation}\label{eqPairing1}
\la\alpha,\beta\ra=\tint_M\alpha\wedge\beta\,,
\end{equation}
where $\alpha\in\f^k(M)$ and $\beta\in\f^{m-k}(M)$ are supposed to have supports with compact intersection. 
In particular, one has a non-degenerate pairing $\la\cdot,\cdot\ra:\fc^k(M)\times\f^{m-k}(M)\to\bbR$. 

If $M$ is endowed with a metric and an orientation 
(for example, this is the case for a spacetime according to Definition\ \ref{defSpacetime}), 
the {\em Hodge star} operator $\ast:\bw^k(T^\ast M)\to\bw^{m-k}(T^\ast M)$ is defined, 
see\ e.g.\ \cite[Section\ 3.3]{Jos11}. $\ast$ induces a non-denerate inner product on $\bw^k(T^\ast M)$ defined by 
\begin{equation}\label{eqContraction}
(\xi,\eta)\in\bw^k(T^\ast M)\times\bw^k(T^\ast M)\mapsto\ast^{-1}(\xi\wedge\ast\eta)\in\bbR\,,
\end{equation} 
while $\ast1\in\f^m(M)$ defines the standard {\em volume form} $\vol$ for a pseudo-Riemannian oriented manifold. 
We get a symmetric pairing $(\cdot,\cdot)$ between $k$-forms 
integrating the inner product defined in eq.\ \eqref{eqContraction} with the volume form $\vol$: 
\begin{equation}\label{eqPairing2}
(\alpha,\beta)=\tint_M\ast^{-1}(\alpha\wedge\ast\beta)\vol=\tint_M\alpha\wedge\ast\beta
=\la\alpha,\ast\beta\ra\,,
\end{equation}
where $\alpha,\beta\in\f^k(M)$ have supports with compact intersection. 
In particular, $(\cdot,\cdot)$ is a non-degenerate bilinear pairing between $\fc^k(M)$ and $\f^k(M)$. 

For differential forms of any degree on a $m$-dimensional manifold $M$, 
one has the {\em differential} $\dd:\f^k(M)\to\f^{k+1}(M)$ forming the well-known {\em de Rham complex}: 
\begin{equation}\label{eqddComplex}
0\lra\f^0(M)\overset{\dd}{\lra}\f^1(M)\overset{\dd}{\lra}\cdots\overset{\dd}{\lra}\f^m(M)\lra0\,.
\end{equation}
From this complex, for each $k\in\{0,\dots,m\}$, one defines {\em de Rham cohomology groups} 
$\hdd^k(M)=\fdd^k(M)/\dd\f^{k-1}(M)$ taking the quotient between the kernel 
$\fdd^k(M)=\ker\big(\dd:\f^k(M)\to\f^{k+1}(M)\big)$ of $\dd$ and its image $\dd\f^{k-1}(M)$. 
One has a similar complex when restricting to compact supports, 
giving rise to {\em compactly supported de Rham cohomology groups} 
$\hcdd^k(M)=\fcdd^k(M)/\dd\fc^{k-1}(M)$, $k\in\{0,\dots,m\}$, 
where $\fcdd^k(M)=\ker\big(\dd:\fc^k(M)\to\fc^{k+1}(M)\big)$. 

A major property of the differential $\dd$ in relation to integration 
on an oriented manifold $M$ is the well-known {\em Stokes' theorem}. 

\begin{theorem}\label{thmStokes}
Let $M$ be a $m$-dimensional oriented manifold (without boundary) and $\omega\in\fc^{m-1}(M)$. 
Then $\int_M\dd\omega=0$. 
\end{theorem}

When $M$ is endowed with a metric and an orientation, 
one can use the Hodge star to introduce another differential operator, 
the {\em codifferential} $\de=(-1)^k\ast^{-1}\dd\ast:\f^k(M)\to\f^{k-1}(M)$. 
Using Stokes' theorem, one can easily show that $\de$ is the formal adjoint of $\dd$, 
that is to say $(\de\alpha,\beta)=(\alpha,\dd\beta)$ 
for each $\alpha\in\f^{k+1}(M)$ and $\beta\in\f^k(M)$ with $\supp(\alpha)\cap\supp(\beta)$ compact. 
Trivially $\de$ gives rise to a second complex (with decreasing degree), 
which is isomorphic to \eqref{eqddComplex} via $\ast$: 
\begin{equation}\label{eqdeComplex}
0\lra\f^m(M)\overset{\de}{\lra}\f^{m-1}(M)\overset{\de}{\lra}\cdots\overset{\de}{\lra}\f^0(M)\lra0\,.
\end{equation}
Defining $\fde^k(M)$ as the kernel of $\de:\f^k(M)\to\f^{k-1}(M)$ 
and $\fcde^k(M)$ as the kernel of $\de:\fc^k(M)\to\fc^{k-1}(M)$, 
one gets cohomology groups $\hde^k(M)=\fde^k(M)/\de\f^{k+1}(M)$ 
and cohomology groups with compact support $\hcde^k(M)=\fcde^k(M)/\de\fc^{k+1}(M)$, $k\in\{0,\dots,m\}$. 
It is easy to check that the Hodge star operator provides an isomorphism
between the cohomology groups for $\de$ and those defined for $\dd$, 
namely $\hdd^k(M)\simeq\hde^{m-k}(M)$ and $\hcdd^k(M)\simeq\hcde^{m-k}(M)$. 

Theorem\ \ref{thmStokes} entails that the pairings $\la\cdot,\cdot\ra$ (for oriented manifolds) 
and $(\cdot,\cdot)$ (for oriented pseudo-Riemannian manifolds) induce pairings between cohomology groups: 
\begin{subequations}\label{eqPoincarePairing}
\begin{align}
\la\cdot,\cdot\ra & :\hcdd^k(M)\times\hdd^{m-k}(M)\to\bbR\,,\\
{}_\de(\cdot,\cdot) & :\hcde^k(M)\times\hdd^k(M)\to\bbR\,,\\
(\cdot,\cdot)_\de & :\hcdd^k(M)\times\hde^k(M)\to\bbR\,.
\end{align}
\end{subequations}
We conclude the present section stating a version of Poincar\'e duality 
for de Rham cohomologies, see\ \cite[Section\ 1.5]{BT82}. 
This requires the notion of a {\em good cover} for a $m$-dimensional manifold $M$, 
namely a cover by open sets such that the intersection of the elements of any finite subset of the cover 
is either empty or diffeomorphic to $\bbR^m$. 

\begin{theorem}\label{thmPoincareDuality}
Let $M$ be an oriented manifold which admits a finite good cover. 
Then the pairing $\la\cdot,\cdot\ra$ between $\hcdd^k(M)$ and $\hdd^{m-k}(M)$ is non-degenerate. 

If $M$ is also endowed with a metric, 
the pairing ${}_\de(\cdot,\cdot)$ between $\hcde^k(M)$ and $\hdd^k(M)$ 
is non-degenerate, as well as the pairing $(\cdot,\cdot)_\de$ between $\hcdd^k(M)$ and $\hde^k(M)$. 
\end{theorem}

The second part of the statement easily follows from the first one taking into account that 
the Hodge star induces isomorphisms $\hcde^k(M)\simeq\hcdd^{m-k}(M)$ and $\hde^k(M)\simeq\hdd^{m-k}(M)$ 
and the pairing $(\cdot,\cdot)$ for forms is symmetric, while $\la\cdot,\cdot\ra$ is graded-symmetric. 
Specifically one has 
\begin{align*}
{}_\de([\alpha],[\beta]) & =(-1)^{k(m-k)}\la[\ast\alpha],[\beta]\ra\,,\quad
& \forall\alpha\in\fcde^k(M),\beta\in\fdd^k(M)\,,\\
([\alpha],[\beta])_\de & =\la[\alpha],[\ast\beta]\ra\,,\quad & \forall\alpha\in\fcdd^k(M),\beta\in\fde^k(M)\,.
\end{align*}

\begin{remark}\label{remPoincareDualityImproved}
Actually one has a bit more than Theorem\ \ref{thmPoincareDuality}. 
Specifically, even if $M$ does not admit a finite good cover, 
the map defined below is an isomorphism of vector spaces, see\ \cite[Section\ V.4]{GHV72}: 
\begin{equation*}
\hdd^{m-k}(M)\to\big(\hcdd^k(M)\big)^\ast\,,\quad[\beta]\mapsto\la\cdot,[\beta]\ra\,.
\end{equation*}
When $M$ does not admit a finite good cover, what might happen is the failure of 
$[\alpha]\in\hcdd^k(M)\mapsto\la[\alpha],\cdot\ra\in\big(\hdd^{m-k}(M)\big)^\ast$ 
being an isomorphism, see\ \cite[Remark 5.7]{BT82}. 
Therefore the pairing \eqref{eqPoincarePairing} is non-degenerate in the second argument 
for any oriented manifold $M$ and this is also the case for the first argument if $M$ has a finite good cover. 

Similar conclusions hold true for $_\de(\cdot,\cdot)$ and $(\cdot,\cdot)_\de$ as well. 
\end{remark}

\section{Spacelike compact cohomology}\label{secSCCohomology}
This section is devoted to construct an isomorphism between 
a special type of cohomology with spacelike compact support for a globally hyperbolic spacetime 
and the cohomology with compact support of a spacelike Cauchy surface. 

For a globally hyperbolic spacetime $M$, 
we denote the space of $k$-forms with spacelike compact support with $\fsc^k(M)$. 
Since both $\dd$ and $\de$ preserve supports, 
one has complexes similar to \eqref{eqddComplex} and \eqref{eqdeComplex} with restricted supports, 
which have to be spacelike compact. This gives rise to new cohomology groups. 

\begin{definition}
Let $M$ be a globally hyperbolic spacetime. We define {\em spacelike compact cohomology groups} according to 
\begin{equation*}
\hscdd^k(M)=\frac{\fscdd^k(M)}{\dd\fsc^{k-1}(M)}\,,
\quad\hscde^k(M)=\frac{\fscde^k(M)}{\de\fsc^{k+1}(M)}\,,
\end{equation*}
where $\fscdd^k(M)=\ker\big(\dd:\fsc^k(M)\to\fsc^{k+1}(M)\big)$, 
while $\fscde^k(M)=\ker\big(\de:\fsc^k(M)\to\fsc^{k-1}(M)\big)$. 
\end{definition}

Also in this case the Hodge star $\ast$ provides an isomorphism between $\hscdd^k(M)$ and $\hscde^{m-k}(M)$. 

To show that $\hscdd^k(M)\simeq\hcdd^{k}(\Sigma)$, $\Sigma$ being a spacelike Cauchy surface for $M$, 
we adopt the following strategy: 
\begin{enumerate}
\item[1.] Theorem\ \ref{thmGlobHyp} ensures that any globally hyperbolic spacetime $M$ is isometric 
to a globally hyperbolic spacetime $M_\Sigma$, whose underlying manifold is $\bbR\times\Sigma$, 
$\Sigma$ being a spacelike Cauchy surface for $M$. 
Therefore one gets $\hscdd^k(M)\simeq\hscdd^k(M_\Sigma)$; 
\item[2.] Exploiting the fact that $M_\Sigma$ is explicitly factored as $\bbR\times\Sigma$, 
we construct an isomorphism $\hscdd^k(M_\Sigma)\simeq\hcdd^{k}(\Sigma)$. 
\end{enumerate}

\begin{remark}
For $k=m$ the isomorphism in item 2 means that $\hscdd^m(M_\Sigma)$ is trivial 
since $\Sigma$ is a manifold of dimension $m-1$. 

As for the isomorphism in item 1, it depends on the specific choice of a foliation $M_\Sigma$ of $M$, 
yet all spacelike Cauchy surfaces of $M$ are diffeomorphic, thus possessing isomorphic cohomology groups. 
\end{remark}

Let us consider a globally hyperbolic spacetime $M$ and a foliation of the type 
provided by Theorem\ \ref{thmGlobHyp}. As above, it is convenient to denote such foliation with $M_\Sigma$. 
Recalling that, as a manifold, $M_\Sigma$ has the form $\bbR\times\Sigma$, we define projections on each factor: 
\begin{equation}\label{eqProj}
t:M_\Sigma\to\bbR\,,\quad\pi:M_\Sigma\to\Sigma\,.
\end{equation}
We introduce also the section $s$ of $\pi$ defined by 
\begin{equation}\label{eqSect}
s:\Sigma\to M_\Sigma\,,\quad x\mapsto(0,x)\,,
\end{equation}
Note that the image of $s$ is the spacelike Cauchy surface $\{0\}\times\Sigma$ of $M_\Sigma$. 
Since the intersection of a spacelike compact region with a spacelike Cauchy surface is always compact 
(refer to \cite[Corollary\ A.5.4]{BGP07}), 
the pullback via $s$ of a spacelike compact form on $M_\Sigma$ is compactly supported on $\Sigma$. 
Furthermore, each compact subset of $\Sigma$ has spacelike compact preimage in $M_\Sigma$ under $\pi$, 
therefore the pullback via $\pi$ of a compactly supported form on $\Sigma$ 
has spacelike compact support on $M_\Sigma$: 
\begin{equation*}
\pi^\ast:\fc^k(\Sigma)\to\fsc^k(M_\Sigma)\,,\quad s^\ast:\fsc^k(M_\Sigma)\to\fc^k(\Sigma)\,.
\end{equation*}
Trivially, $s^\ast\pi^\ast=\id_{\fc^k(\Sigma)}$, hence a similar identity holds true in cohomology. 
If one can prove that $\pi^\ast s^\ast$ is chain homotopic to $\id_{\fsc^\ast(M_\Sigma)}$, 
then $\pi^\ast s^\ast$ induces a map in cohomology which coincides with the identity . 
This result is achieved in Lemma\ \ref{lemSCHomotopy} below. 
In particular, $\pi^\ast$ provides the sought isomorphism $\hcdd^k(\Sigma)\simeq\hscdd^k(M_\Sigma)$. 

The relevant chain homotopy is defined following \cite[Section\ I.4, pp.\ 33--35]{BT82}. 
Due to the slightly different setting considered here, 
we have to be careful with the support properties of the chain homotopy we want to consider. 
Specifically, we will observe that the chain homotopy mentioned above is such that 
any form with spacelike compact support is mapped to a form with spacelike compact support as well. 
Therefore the formula given in \cite{BT82} provides 
an appropriate chain homotopy for spacelike compact cohomologies too. 
The rest of the present section is devoted to a detailed description of this argument. 

First, we note that spacelike compact $k$-forms on $M_\Sigma$ are always given 
by a linear combination of two types of forms:\footnote{From 
now on the wedge product will be often suppressed to lighten the notation.} 
\begin{align}
(\pi^\ast\phi)\,f\,, & \quad\phi\in\f^k(\Sigma),\,f\in\csc(M_\Sigma)\,,\tag{$1_\mathrm{sc}$}\label{eqType1SC}\\
(\pi^\ast\psi)\,h\dd t\,, & \quad\psi\in\f^{k-1}(\Sigma),\,
h\in\csc(M_\Sigma)\tag{$2_\mathrm{sc}$}\label{eqType2SC}\,.
\end{align}

The main point of the present section is to note that, given $f\in\csc(M_\Sigma)$, 
$\tilde{f}:(t,x)\mapsto\int_0^tf(s,x)\dd s$ is still in $\csc(M_\Sigma)$. 
Since the support of $f$ is spacelike compact, there exists a compact subset $K\subseteq M_\Sigma$ 
such that $\supp(f)\subseteq J_{M_\Sigma}(K)$. 
We consider the compact set $\tilde{K}=\big(\{0\}\times\Sigma\big)\cap J_{M_\Sigma}(K)$ 
and we show that $\tilde{f}=0$ outside $J_{M_\Sigma}(\tilde{K})$. 
Let $(t,x)\in M_\Sigma\setminus J_{M_\Sigma}(\tilde{K})$. 
This entails that the timelike curve $s\in[0,t]\mapsto(s,x)\in M_\Sigma$ does not meet $\tilde{K}$. 
By construction $J_{M_\Sigma}(K)\subseteq J_{M_\Sigma}(\tilde{K})$, 
therefore $f(s,x)=0$ for each $s\in[0,t]$ and hence $\tilde{f}(t,x)=0$.  
 
With this fact in mind, one defines 
\begin{align*}
P:\fsc^k(M_\Sigma) & \to\fsc^{k-1}(M_\Sigma)\,,\\
(\pi^\ast\phi)\,f & \mapsto0\,,\\
(\pi^\ast\psi)\,h\dd t & \mapsto(\pi^\ast\psi)\,\tint_0^\cdot h(s,\cdot)\dd s\,.
\end{align*}

\begin{lemma}\label{lemSCHomotopy}
Let $P:\fsc^k(M_\Sigma)\to\fsc^{k-1}(M_\Sigma)$ be defined as above. 
Then $P$ is a {\em chain homotopy operator} between $\pi^\ast s^\ast$ and $\id_{\fsc^\ast(M_\Sigma)}$, 
that is to say $\pi^\ast s^\ast-\id_{\fsc^k(M_\Sigma)}=(-1)^k(\dd\,P-P\,\dd)$ 
on $\fsc^k(M_\Sigma)$ for each $k\in\{0,\dots,m\}$. 
\end{lemma}

\begin{proof}
The proof is nothing but a calculation identical to the one in \cite[Section\ I.4, pp.\ 33--35]{BT82}. 
We give a sketch for the sake of completeness. 

The calculation is performed choosing an atlas for $\Sigma$ and extending it to an atlas for $M_\Sigma$. 

We consider first the case of $k$-forms of type \eqref{eqType1SC}: 
\begin{align*}
(\dd\,P-P\,\dd)\big((\pi^\ast\phi)f\big) & =-P\big((\pi^\ast\dd\phi)\,f+(-1)^k(\pi^\ast\phi)\,\dd f\big)\\
& =(-1)^{k+1}(\pi^\ast\phi)\,\tint_0^\cdot\partial_sf(s,\cdot)\dd s\\
& =(-1)^{k+1}\big((\pi^\ast\phi)\,f-\pi^\ast\big(\phi\, f(0,\cdot)\big)\big)\\
& =(-1)^k\big(\pi^\ast s^\ast-\id_{\fsc^k(M_\Sigma)}\big)\big((\pi^\ast\phi)\,f\big)\,.
\end{align*}

The computation for $k$-forms of type \eqref{eqType2SC} is a bit more involved. 
As a first step, we compute the $\dd\,P$-term: 
\begin{align*}
\dd\,P\big((\pi^\ast\psi)\,h\dd t\big) & =\dd\big((\pi^\ast\psi)\,\tint_0^\cdot h(s,\cdot)\dd s\big)\\
& =(\pi^\ast\dd\psi)\,\tint_0^\cdot h(s,\cdot)\dd s\\
& \quad+(-1)^{k-1}(\pi^\ast\psi)\,\big(\dd x^i\partial_i\tint_0^\cdot h(s,\cdot)\dd s+h\dd t\big)\,.
\end{align*}
Then we focus our attention on the term involving $P\,\dd$: 
\begin{align*}
P\,\dd\big((\pi^\ast\psi)\,h\dd t\big) & 
= P\,\big((\pi^\ast\dd\psi)\,h\dd t+(-1)^{k-1}(\pi^\ast\psi)\,\dd x^i\,\partial_ih\dd t\big)\\
& =(\pi^\ast\dd\psi)\,\tint_0^\cdot h(s,\cdot)\dd s
+(-1)^{k-1}(\pi^\ast\psi)\,\dd x^i\,\tint_0^\cdot\partial_ih(s,\cdot)\dd s\,.
\end{align*}
Subtracting the $P\,\dd$-term from the $\dd\,P$-term, one completes the proof: 
\begin{align*}
(\dd\,P-P\,\dd)\big((\pi^\ast\psi)\,h\dd t\big) & =(-1)^{k-1}(\pi^\ast\psi)\,h\dd t\\
& =(-1)^k(\pi^\ast s^\ast-\id_{\fsc^k(M_\Sigma)})\big((\pi^\ast\psi)\,h\dd t\big)\,.
\end{align*}
Note that we exploited the possibility to exchange the integral in time 
and the spatial derivative to obtain the first equality, 
while the second is a consequence of $s^\ast\dd t=0$. 
\end{proof}

The last lemma enables us to prove the main result of this section. 

\begin{theorem}\label{thmSCCohomology}
Let $M$ be a globally hyperbolic spacetime and consider a spacelike Cauchy surface $\Sigma$ for $M$. 
Then $\pi$ and $s$, defined respectively in \eqref{eqProj} and \eqref{eqSect}, 
induce isomorphisms in cohomology: 
\begin{equation*}
\xymatrix{\hscdd^\ast(M)\ar@/^1.5pc/[r]^{s^\ast} & \hcdd^\ast(\Sigma)\ar@/^1.5pc/[l]^{\pi^\ast}}\,.
\end{equation*}
\end{theorem}

\begin{proof}
Theorem\ \ref{thmGlobHyp} entails that $M$ is isometric to the globally hyperbolic spacetime $M_\Sigma$ 
obtained endowing $\bbR\times\Sigma$ with suitable metric, orientation and time-orientation, 
hence $\hscdd^\ast(M)\simeq\hscdd^\ast(M_\Sigma)$. 
On a side we have $s^\ast\pi^\ast=\id_{\fc^k(\Sigma)}$. 
On the other side Lemma\ \ref{lemSCHomotopy} provides the identity 
$\pi^\ast s^\ast-\id_{\fsc^k(M_\Sigma)}=(-1)^k(\dd\,P-P\,\dd)$ on $\fsc^k(M_\Sigma)$. 
Since the term on the right hand side maps $\fscdd^k(M_\Sigma)$ to $\dd\fsc^{k-1}(M_\Sigma)$, 
$\pi^\ast s^\ast$ induces the identity in cohomology, thus concluding the proof. 
\end{proof}

\begin{example}\label{exaSCCohomology}
There are several physically relevant examples of globally hyperbolic spacetimes 
with non trivial spacelike compact cohomology groups. This fact is made evident 
by Theorem\ \ref{thmSCCohomology}, which enables us to compute spacelike compact cohomologies 
just looking at the cohomology groups with compact support of a spacelike Cauchy surface. 
Some examples in dimension $m=4$ are:
\begin{description}
\item[Einstein's static universe] Any spacelike Cauchy surface is diffeomorphic to $\bbS^3$, 
therefore $\hscdd^\ast(M)\simeq(\bbR,0,0,\bbR,0)$; 
\item[Schwarzschild spacetime] Any spacelike Cauchy surface is diffeomorphic to $\bbR\times\bbS^2$, 
therefore $\hscdd^\ast(M)\simeq(0,\bbR,0,\bbR,0)$; 
\item[Gowdy's $\bbT^3$ spacetime \cite{Gow74}] Any spacelike Cauchy surface is diffeomorphic to $\bbT^3$, 
therefore $\hscdd^\ast(M)\simeq(\bbR,\bbR^3,\bbR^3,\bbR,0)$; 
\end{description}
where $\bbS^n$ denotes the $n$-sphere and $\bbT^n$ denotes the $n$-torus. 
\end{example}

\section{Timelike compact cohomology}\label{secTCCohomology}
In this section we consider cohomology with timelike compact support on a globally hyperbolic spacetime 
and we show that it is isomorphic to the de Rham cohomology of a spacelike Cauchy surface 
(with degree lowered by $1$). 

Similarly to Section\ \ref{secSCCohomology}, we define the space $\ftc^k(M)$ 
of $k$-forms with timelike compact support on a globally hyperbolic spacetime $M$. 
The usual observation that both $\dd$ and $\de$ preserve supports entails the existence of complexes 
similar to \eqref{eqddComplex} and \eqref{eqdeComplex}, but restricted to forms with timelike compact support. 
This suggests the definition of timelike compact cohomology groups. 

\begin{definition}
Let $M$ be a globally hyperbolic spacetime. We define {\em timelike compact cohomology groups} according to 
\begin{equation*}
\htcdd^k(M)=\frac{\ftcdd^k(M)}{\dd\ftc^{k-1}(M)}\,,\quad\htcde^k(M)=\frac{\ftcde^k(M)}{\de\ftc^{k+1}(M)}\,,
\end{equation*}
where $\ftcdd^k(M)=\ker\big(\dd:\ftc^k(M)\to\ftc^{k+1}(M)\big)$ 
and $\ftcde^k(M)=\ker\big(\de:\ftc^k(M)\to\ftc^{k-1}(M)\big)$. 
\end{definition}

As always, $\htcdd^k(M)$ and $\htcde^{m-k}(M)$ are isomorphic via the Hodge star operator $\ast$. 

The strategy to prove the isomorphism with de Rham cohomology of a spacelike Cauchy surface
(with degree lowered by $1$) is the following: 
\begin{enumerate}
\item[1.] As in the previous section, for each globally hyperbolic spacetime $M$, 
we exploit Theorem\ \ref{thmGlobHyp} to show that it is isometric to a globally hyperbolic spacetime $M_\Sigma$, 
whose underlying manifold is $\bbR\times\Sigma$, $\Sigma$ being a spacelike Cauchy surface for $M$. 
Therefore one gets $\htcdd^k(M)\simeq\htcdd^k(M_\Sigma)$; 
\item[2.] Exploiting the fact that $M_\Sigma$ has $\bbR\times\Sigma$ as underlying manifold, 
we construct an isomorphism $\htcdd^k(M_\Sigma)\simeq\hdd^{k-1}(\Sigma)$. 
\end{enumerate}

\begin{remark}
It is understood that, for $k=0$, the isomorphism in item 2 means that $\htcdd^0(\bbR\times\Sigma)$ is trivial. 
This can be seen directly by noting that a constant function with timelike compact support must vanish everywhere. 
\end{remark}

Take a globally hyperbolic spacetime $M$ and consider a foliation $M_\Sigma$ 
provided by Theorem\ \ref{thmGlobHyp}. In particular, note that the manifold underlying $M_\Sigma$ 
is the Cartesian product $\bbR\times\Sigma$, $\Sigma$ being a spacelike Cauchy surface for $M$.  
Again, we consider the projections $\pi$ and $t$ defined in \eqref{eqProj}. 

The maps giving rise to isomorphisms in cohomology are defined similarly to \cite[Section\ I.6, pp.\ 61--63]{BT82}, 
the only difference being a stricter constraint on the support of forms (timelike compact implies 
vertical compact as defined in \cite{BT82}, but the converse does not hold). 
Therefore, we only have to check that the relevant maps are well-behaved 
with respect to timelike compact supports. For the sake of completeness, 
we recall below the main argument of \cite{BT82} adapted to the present setting. 

We start from the observation that all $k$-forms with timelike compact support on $M_\Sigma$ 
are linear combinations of two types of forms:\footnote{As in the previous section, 
we will often omit the wedge product in order to make the notation more readable.} 
\begin{align}
(\pi^\ast\phi)\,f\,, & \quad\phi\in\f^k(\Sigma),\,f\in\ctc(M_\Sigma)\,,\tag{$1_\mathrm{tc}$}\label{eqType1TC}\\
(\pi^\ast\psi)\,h\dd t\,, & 
\quad\psi\in\f^{k-1}(\Sigma),\,h\in\ctc(M_\Sigma)\tag{$2_\mathrm{tc}$}\label{eqType2TC}\,.
\end{align}
 
Since for each compact subset $K$ of $\Sigma$, 
$\pi^{-1}(K)$ has compact intersection with a timelike compact region $T$ of $M_\Sigma$ 
(the intersection being included in the compact set $J_{M_\Sigma}(K)\cap T$), 
we can introduce the time-integration map: 
\begin{align}\label{eqTimeInt}
i:\ftc^k(M_\Sigma) & \to\f^{k-1}(\Sigma)\,,\\
(\pi^\ast\phi)\,f & \mapsto0\,,\nonumber\\
(\pi^\ast\psi)\,h\dd t & \mapsto\psi\,\tint_\bbR h(s,\cdot)\dd s\,.\nonumber
\end{align}
In order to show that $i$ descends to cohomologies, 
one checks that it is a morphisms between the complexes $\big(\ftc^\ast(M_\Sigma),\dd\big)$ 
and $\big(\f^{\ast-1}(\Sigma),\dd\big)$. This result is shown in the following lemma. 

\begin{lemma}
$\dd\,i=i\,\dd$ on $\ftc^k(M_\Sigma)$ for each $k\in\{0,\dots,m\}$. 
\end{lemma}

\begin{proof}
The proof is a straightforward calculation performed choosing an oriented atlas for $\Sigma$ 
and extending it to an atlas for $M_\Sigma$. 
One has only to take into account the possibility to exchange spatial derivatives and integrals in time, 
which follows from $\pi^{-1}(K)$, $K\subseteq\Sigma$ compact, 
having compact intersection with each timelike compact region of $M_\Sigma$. 

First we consider $k$-forms of type \eqref{eqType1TC}: 
\begin{align*}
i\,\dd\big((\pi^\ast\phi)\,f\big) & =i\big((\pi^\ast\dd\phi)\,f+(-1^k)(\pi^\ast\phi)\,\dd f\big)\\
& =(-1^k)\phi\,\tint_\bbR\partial_sf(s,\cdot)\dd s=0=\dd\,i\big((\pi^\ast\phi)\,f\big)\,,
\end{align*}
where $\tint_\bbR\partial_sf(s,x)\dd s=0$ for each $x\in\Sigma$ 
due to the fact that $s\mapsto f(s,x)$ has compact support in $\bbR$. 

To complete the proof, we consider $k$-forms of type \eqref{eqType2TC}: 
\begin{align*}
\dd\,i\big((\pi^\ast\psi)\,h\dd t\big) & =\dd\big(\psi\,\tint_\bbR h(s,\cdot)\dd s\big)\\
& =\dd\psi\,\tint_\bbR h(s,\cdot)\dd s+(-1)^{k-1}\psi\,\dd x^i\,\partial_i\tint_\bbR h(s,\cdot)\dd s\,,\\
i\,\dd\big((\pi^\ast\psi)\,h\dd t\big) & 
=i\,\big((\pi^\ast\dd\psi)\,h\dd t+(-1)^{k-1}(\pi^\ast\psi)\,\dd x^i\,\partial_ih\dd t\big)\\
& =\dd\psi\,\tint_\bbR h(s,\cdot)\dd s+(-1)^{k-1}\psi\,\dd x^i\,\tint_\bbR\partial_ih(s,\cdot)\dd s\,.
\end{align*}
As noted above, it is possible to reverse the order 
in which the integral in time and the spatial derivative are performed, therefore the thesis follows. 
\end{proof}

$(k-1)$-forms on $\Sigma$ are extended to timelike compact $k$-forms on $M_\Sigma$ 
taking the wedge product with a suitable $1$-form: Let $a\in\cc(\bbR)$ be such that $\int_\bbR a(s)\dd s=1$. 
Consider the $1$-form $\omega=(t^\ast a)\dd t$ on $M_\Sigma$ 
and observe that it is closed and timelike compact. That done, define the time-extension map: 
\begin{equation}\label{eqTimeExt}
e:\f^{k-1}(\Sigma)\to\ftc^k(M_\Sigma)\,,\quad\phi\mapsto(\pi^\ast\phi)\wedge\omega\,.
\end{equation}
$e:\f^{\ast-1}(\Sigma)\to\ftc^\ast(M_\Sigma)$ is a morphism of complexes due to the identity 
\begin{equation*}
\dd\,e\,\phi=\dd\big((\pi^\ast\phi)\wedge\omega\big)=(\pi^\ast\dd\phi)\wedge\omega=e\,\dd\,\phi\,,
\end{equation*}
for each $k\in\{1,\dots,m+1\}$ and $\phi\in\f^{k-1}(\Sigma)$. 

From \eqref{eqTimeInt} and \eqref{eqTimeExt}, 
one immidiately realizes that $i\,e=\id_{\f^{\ast-1}(\Sigma)}$, which entails a similar identity in cohomology: 
\begin{equation}\label{eqIntExt}
i\,e\,\phi=i\big((\pi^\ast\phi)\wedge\omega\big)=\phi\tint_\bbR a(s)\dd s=\phi\,,
\end{equation}
for each $k\in\{1,\dots,m+1\}$ and $\phi\in\f^{k-1}(\Sigma)$. 
To conclude that $e$ and $i$ induce the sought isomorphisms in cohomology, 
one has to find a chain homotopy $Q:\ftc^k(M_\Sigma)\to\ftc^{k-1}(M_\Sigma)$ 
between $e\,i$ and $\id_{\ftc^\ast(M_\Sigma)}$. 

We begin noting the following fact, which is the key point to show that the homotopy operator 
defined in \cite[Section\ I.4, p.\ 38]{BT82} maps timelike compact $k$-forms to $\ftc^{k-1}(M_\Sigma)$: 
Given $f\in\ctc(M_\Sigma)$, the function $\widehat{f}:M_\Sigma\to\bbR$, defined by the formula 
\begin{equation}\label{eqHat}
\widehat{f}(t,x)=\tint_{-\infty}^tf(s,x)\dd s-\tint_\bbR f(r,x)\dd r\,\tint_{-\infty}^ta(s)\dd s\,,
\end{equation}
lies in $\ctc(M_\Sigma)$. In particular, one has to check that the support of $\widehat{f}$ is timelike compact: 
Per construction $a$ is supported in a time slab 
$[c,d]\times\Sigma=J_{M_\Sigma}^+(\Sigma_c)\cap J_{M_\Sigma}^-(\Sigma_d)$ for $c,d\in\bbR$, 
where $\Sigma_c=\{c\}\times\Sigma$ and $\Sigma_d=\{d\}\times\Sigma$. 
According to \cite[Theorem 3.1]{San13}, we also find Cauchy surfaces $\tilde{\Sigma}_\pm$ 
such that $\supp(f)\subseteq J_{M_\Sigma}^+(\tilde{\Sigma}_-)\cap J_{M_\Sigma}^-(\tilde{\Sigma}_+)$. 
As one can realize looking at \eqref{eqHat}, 
$\widehat{f}=0$ in $I_{M_\Sigma}^+(\Sigma_d)\cap I_{M_\Sigma}^+(\tilde{\Sigma}_+)$ 
and in $I_{M_\Sigma}^-(\Sigma_c)\cap I_{M_\Sigma}^-(\tilde{\Sigma}_-)$. 
As a matter of fact, for the time integral of $a$ one has 
\begin{align*}
\tint_{-\infty}^ta(s)\dd s=
\begin{cases}
1\,, & (t,x)\in I_{M_\Sigma}^+(\Sigma_d)\,,\\
0\,, & (t,x)\in I_{M_\Sigma}^-(\Sigma_c)\,,
\end{cases}
\end{align*}
while the time integral of $f$ gives 
\begin{align*}
\tint_{-\infty}^tf(s,x)\dd s=
\begin{cases}
\tint_\bbR f(s,x)\dd s\,, & (t,x)\in I_{M_\Sigma}^+(\tilde{\Sigma}_+)\,,\\
0\,, & (t,x)\in I_{M_\Sigma}^-(\tilde{\Sigma}_-)\,.
\end{cases}
\end{align*}
This entails that $\widehat{f}$ is supported inside the intersection between the causal future of both $\Sigma_c$ 
and $\tilde{\Sigma}_-$ and the causal past of both $\Sigma_d$ and $\tilde{\Sigma}_+$,\footnote{This 
is the complement in $M_\Sigma$ of the union between 
the chronological future of both $\Sigma_d$ and $\tilde{\Sigma}_+$ 
and the chronological past of both $\Sigma_c$ and $\tilde{\Sigma}_-$. 
In fact, each Cauchy surface $\Sigma$ for a globally hyperbolic spacetime $M$ splits it into two disjoint parts 
according to $I_M^\pm(\Sigma)\cup J_M^\mp(\Sigma)=M$. 
This fact easily follows from the definition of Cauchy surface.} 
namely 
\begin{equation*}
\supp(\widehat{f})\subseteq\big(J_{M_\Sigma}^+(\Sigma_c)\cup J_{M_\Sigma}^+(\tilde{\Sigma}_-)\big)
\cap\big(J_{M_\Sigma}^-(\Sigma_d)\cap J_{M_\Sigma}^-(\tilde{\Sigma}_+)\big)\,. 
\end{equation*}
We deduce that $\supp(\widehat{f})$ is timelike compact. 

This enables us to define the linear map below, which provides the sought homotopy operator: 
\begin{align*}
Q:\ftc^k(M_\Sigma) & \to\ftc^{k-1}(M_\Sigma)\,,\\
(\pi^\ast\phi)\,f & \mapsto0\,,\\
(\pi^\ast\psi)\,h\dd t & \mapsto(\pi^\ast\psi)\,\widehat{h}\,,
\end{align*}
where $\widehat{h}$ is defined according to eq.\ \eqref{eqHat}. 

\begin{lemma}\label{lemTCHomotopy}
Let $Q:\ftc^k(M_\Sigma)\to\ftc^{k-1}(M_\Sigma)$ be defined as above. 
Then $Q$ is a chain homotopy between $e\,i$ and $\id_{\ftc^\ast(M_\Sigma)}$, that is to say 
$e\,i-\id_{\ftc^\ast(M_\Sigma)}=(-1)^k(\dd\,Q-Q\,\dd)$ on $\ftc^k(M_\Sigma)$ for each $k\in\{0,\dots,m\}$. 
\end{lemma}

\begin{proof}
As always, the computation is performed fixing an oriented atlas for $\Sigma$ and extending it to $M_\Sigma$. 

We consider first the case of a $k$-form of type \eqref{eqType1TC}: 
\begin{align*}
(\dd\,Q-Q\,\dd)\big((\pi^\ast\phi)f\big) & =-Q\big((\pi^\ast\dd\phi)\,f+(-1)^k(\pi^\ast\phi)\,\dd f\big)\\
& =(-1)^{k+1}(\pi^\ast\phi)\,\widehat{\partial_tf}
=(-1)^k(e\,i-\id_{\ftc^\ast(M_\Sigma)})\big((\pi^\ast\phi)\,f\big)\,,
\end{align*}
where the last equality follows from the definition of $i$, eq.\ \eqref{eqTimeInt}, 
and the identity $\widehat{\partial_tf}=f$, which in turn is a consequence of eq.\ \eqref{eqHat} 
on account of the compact support in $\bbR$ of the function $t\mapsto f(t,x)$ for each fixed $x\in\Sigma$. 

Considering the case of $k$-forms of type \eqref{eqType2TC}, one has the following: 
\begin{align*}
\dd\,Q\big((\pi^\ast\psi)\,h\dd t\big) & =\dd\big((\pi^\ast\psi)\,\widehat{h}\big)
=(\pi^\ast\dd\psi)\,\widehat{h}+(-1)^{k-1}(\pi^\ast\psi)\,\dd\widehat{h}\,,\\
Q\,\dd\big((\pi^\ast\psi)\,h\dd t\big) & 
=Q\big((\pi^\ast\dd\psi)\,h\dd t+(-1)^{k-1}(\pi^\ast\psi)\,\dd x^i\,\partial_ih\dd t\big)\\
& =(\pi^\ast\dd\psi)\,\widehat{h}+(-1)^{k-1}(\pi^\ast\psi)\,\dd x^i\,\widehat{\partial_ih}\,,\\
e\,i\big((\pi^\ast\psi)h\dd t\big) & =e\big(\psi\,\tint_\bbR h(s,\cdot)\dd s\big)
=\big(\pi^\ast\big(\psi\,\tint_\bbR h(s,\cdot)\dd s\big)\big)\,\omega\,.
\end{align*}
From eq.\ \eqref{eqHat} one reads 
$\dd\widehat{h}=h\dd t-\pi^\ast\big(\tint_\bbR h(r,\cdot)\dd r\big)\,\omega+\dd x^i\,\partial_i\widehat{h}$, 
therefore 
\begin{align*}
(\dd\,Q-Q\,\dd)\big((\pi^\ast\psi)\,h\dd t\big) & =(-1)^{k-1}(\pi^\ast\psi)\,h\dd t
-(-1)^{k-1}\big(\pi^\ast\big(\psi\,\tint_\bbR h(r,\cdot)\dd r\big)\big)\,\omega\\
& =(-1)^k\big(e\,i-\id_{\ftc^k(M_\Sigma)}\big)\big((\pi^\ast\psi)\,h\dd t\big)\,,
\end{align*}
concluding the proof. 
\end{proof}

The last lemma leads us to the main result of the present section. 

\begin{theorem}\label{thmTCCohomology}
Let $M$ be a globally hyperbolic spacetime and consider a spacelike Cauchy surface $\Sigma$ for $M$. 
Then $i$ and $e$, defined respectively in \eqref{eqTimeInt} and \eqref{eqTimeExt}, 
induce isomorphisms in cohomology: 
\begin{equation*}
\xymatrix{\htcdd^\ast(M)\ar@/^1.5pc/[r]^{i} & \hdd^{\ast-1}(\Sigma)\ar@/^1.5pc/[l]^{e}}\,.
\end{equation*}
\end{theorem}

\begin{proof}
From Theorem\ \ref{thmGlobHyp} one reads that $M$ is isometric to the globally hyperbolic spacetime $M_\Sigma$, 
the underlying manifold being $\bbR\times\Sigma$ endowed with suitable metric, orientation and time-orientation. 
In particular, this entails $\htcdd^\ast(M)\simeq\htcdd^\ast(M_\Sigma)$. 
Similarly to the proof of Theorem\ \ref{thmSCCohomology}, 
eq.\ \eqref{eqIntExt} and Lemma\ \ref{lemTCHomotopy} show 
that $i$ and $e$ give rise to the sought isomorphisms in cohomology, thus completing the proof. 
\end{proof}

\begin{example}
As in Example\ \ref{exaSCCohomology}, one can exploit Theorem\ \ref{thmTCCohomology} 
to find examples of globally hyperbolic spacetimes with non-trivial timelike compact cohomology groups. 
Considering the same spacetimes of the example mentioned above, one has:
\begin{description}
\item[Einstein's static universe] Any spacelike Cauchy surface is diffeomorphic to $\bbS^3$, 
therefore $\htcdd^\ast(M)\simeq(0,\bbR,0,0,\bbR)$; 
\item[Schwarzschild spacetime] Any spacelike Cauchy surface is diffeomorphic to $\bbR\times\bbS^2$, 
therefore $\htcdd^\ast(M)\simeq(0,\bbR,0,\bbR,0)$; 
\item[Gowdy's $\bbT^3$ spacetime \cite{Gow74}] Any spacelike Cauchy surface is diffeomorphic to $\bbT^3$, 
therefore $\htcdd^\ast(M)\simeq(0,\bbR,\bbR^3,\bbR^3,\bbR)$. 
\end{description}
\end{example}

\section{Poincar\'e duality between $\hscdd^\ast$ and $\htcdd^{m-\ast}$}\label{secSCTCDuality}
The aim of this section is to extend the usual Poincar\'e duality 
between de Rham cohomology and its counterpart with compact support, see\ Theorem\ \ref{thmPoincareDuality}, 
to the case of cohomologies with spacelike compact and respectively timelike compact support. 

Given a globally hyperbolic spacetime $M$, 
we observe that eq.\ \eqref{eqPairing1} and eq.\ \eqref{eqPairing2} provide pairings 
between spacelike compact and timelike compact forms since by definition spacelike compact regions 
intersect timelike compact ones inside a compact set, see the end of Section\ \ref{secLorGeom}. 
Non-degeneracy in the standard case carries over to the present situation 
since compact forms are both spacelike compact and timelike compact. 

Similarly to the case of \eqref{eqPoincarePairing}, 
one can exploit Theorem\ \ref{thmStokes} to show that $\la\cdot,\cdot\ra$ and $(\cdot,\cdot)$ descend 
to cohomologies with spacelike and timelike compact support. 
Consider for example $\alpha\in\fsc^{k-1}(M)$ and $\beta\in\ftcdd^{m-k}(M)$. 
Then, on account of Theorem\ \ref{thmStokes} and the compact support of $\alpha\wedge\beta$, one has 
\begin{equation*}
\la\dd\alpha,\beta\ra=\tint_M\dd(\alpha\wedge\beta)=0\,.
\end{equation*}
To sum up, we have to show non-degeneracy for the pairings listed below: 
\begin{subequations}\label{eqSCTCPairing}
\begin{align}
\la\cdot,\cdot\ra & :\hscdd^k(M)\times\htcdd^{m-k}(M)\to\bbR\,,\\
{}_\de(\cdot,\cdot) & :\hscde^k(M)\times\htcdd^k(M)\to\bbR\,,\\
(\cdot,\cdot)_\de & :\hscdd^k(M)\times\htcde^k(M)\to\bbR\,.
\end{align}
\end{subequations}
Actually, it is enough to prove non-degeneracy for one of the pairings, 
the others being related via Hodge star $\ast$, see the comment after Theorem\ \ref{thmPoincareDuality}. 

\begin{lemma}\label{lemEquivalence}
Let $M$ be a globally hyperbolic spacetime and consider a spacelike Cauchy surface $\Sigma$ for $M$. 
Denote the isomorphisms provided by Theorem\ \ref{thmSCCohomology} and Theorem\ \ref{thmTCCohomology} 
with $\pi^\ast:\hcdd^\ast(\Sigma)\to\hscdd^\ast(M)$ and $e:\hdd^{m-1-\ast}(\Sigma)\to\htcdd^{m-\ast}(M)$. 
Then $\big\la\pi^\ast[\phi],e\,[\psi]\big\ra=\big\la[\phi],[\psi]\big\ra$ 
for each $k\in\{0,\dots,m-1\}$, $[\phi]\in\hcdd^k(\Sigma)$ and $[\psi]\in\fdd^{m-1-k}(M)$. 
where \eqref{eqSCTCPairing} gives the pairing on the left-hand-side, 
while the one on the right-hand-side is given by \eqref{eqPoincarePairing} for the oriented manifold $\Sigma$. 
\end{lemma}

\begin{proof}
For $k\in\{0,\dots,m-1\}$, take $[\phi]\in\hcdd^k(\Sigma)$ and $[\psi]\in\fdd^{m-k-1}(M)$. 
Recalling \eqref{eqProj} and \eqref{eqTimeExt}, 
one can explicitly compute $\big\la\pi^\ast[\phi],e\,[\psi]\big\ra$: 
\begin{equation*}
\big\la\pi^\ast[\phi],e[\psi]\big\ra=\tint_{M_\Sigma}(\pi^\ast\phi)\wedge\big((\pi^\ast\psi)\wedge\omega\big)
=\big(\int_\Sigma\phi\wedge\psi\big)\tint_\bbR a(s)\dd s=\big\la[\phi],[\psi]\big\ra\,,
\end{equation*}
$M_\Sigma$ being the foliation of $M$ induced by the spacelike Cauchy surface $\Sigma$ 
according to Theorem\ \ref{thmGlobHyp}. 
Note that in the last step we exploited the identity $\tint_\bbR a(s)\dd s=1$. 
\end{proof}

\begin{theorem}\label{thmSCTCPoincareDuality}
Let $M$ be a globally hyperbolic spacetime which admits a finite good cover. 
Then the pairing $\la\cdot,\cdot\ra$ between $\hscdd^k(M)$ and $\htcdd^{m-k}(M)$ is non-degenerate. 
Therefore the same holds true for ${}_\de(\cdot,\cdot)$ between $\hscde^k(M)$ and $\htcdd^k(M)$ 
and $(\cdot,\cdot)_\de$ between $\hscdd^k(M)$ and $\htcde^k(M)$ as well. 
\end{theorem}

\begin{proof}
Consider a spacelike Cauchy surface $\Sigma$ of $M$. 
Lemma\ \ref{lemEquivalence} entails that the pairing between $\hscdd^k(M)$ and $\htcdd^{m-k}(M)$ 
is equivalent to the one between $\hcdd^k(\Sigma)$ and $\hdd^{m-k-1}(\Sigma)$. 
Since $M$ admits a finite good cover, the same holds true for $\Sigma$. 
Therefore non-degeneracy of the pairing between $\hscdd^k(M)$ and $\htcdd^{m-k}(M)$ 
follows from non-degeneracy of the pairing between $\hcdd^k(\Sigma)$ and $\hdd^{m-k-1}(\Sigma)$, 
which holds true on account of Theorem\ \ref{thmPoincareDuality}. 
\end{proof}

\begin{remark}
On account of Remark \ref{remPoincareDualityImproved}, one can obtain a positive result under weaker conditions. 
Even if there is no finite good cover for $M$, one gets an isomorphism 
$\htcdd^{m-k}(M)\to\big(\hscdd^k(M)\big)^\ast$ defined by $[\beta]\mapsto\la\cdot,[\beta]\ra$. 
Similar conclusions hold true also for the other pairings considered in \eqref{eqSCTCPairing}. 
\end{remark}

\section{Classical observables for $k$-form Maxwell fields}\label{secObservables}
In this section we briefly discuss the classical field theory of electromagnetism 
without external sources and its analogs in a different degree. 
We adopt two different viewpoints. In the first place we take the perspective of the vector potential, 
while in the second place we consider the Faraday tensor as the central object. 
The aim is to exhibit models where our knowledge about cohomologies 
with spacelike compact and timelike compact support can be fruitfully exploited in order to better understand 
which space of observables is optimal for the classical field theoretical model under analysis. 
The choice not to include external sources is motivated by the fact that 
the relevant features for the present discussion appear in relation to homogeneous field equation. 
One might take into account sources as well, 
thus dealing with inhomogeneous field equations, see e.g.\ \cite{BDS13a}.

\subsection{The Laplace-de Rham differential operator}
Before discussing the specific models mentioned above, 
it is convenient to recall some well-known properties of the {\em Laplace-de Rham operator} 
$\Box=\de\dd+\dd\de:\f^k(M)\to\f^k(M)$ on a globally hyperbolic spacetime $M$. 
We stress that any value of $k\in\{0,\dots,m\}$ is allowed. 
With a slight abuse of notation, we denote the Laplace-de Rham operator 
with the same symbol regardless of the degree $k$. 

As one can check directly, $\Box$ is {\em normally hyperbolic} and {\em formally self-adjoint}, 
meaning that it is a second order linear differential operator whose principal symbol is given by the metric, 
see\ e.g.\ \cite[Section\ 1.5]{BGP07}, and such that $(\Box\alpha,\beta)=(\alpha,\Box\beta)$ 
for each $\alpha,\beta\in\f^k(M)$ whose supports have compact intersection, 
where $(\cdot,\cdot)$ denotes the pairing defined in \eqref{eqPairing2}. 
This fact easily follows from $\de$ being the formal adjoint of $\dd$, 
see\ Theorem\ \ref{thmStokes} and the following discussion. 

Since $\Box$ is normally hyperbolic and formally self-adjoint, 
there are unique {\em advanced and retarded Green operators} $G_\pm:\fc^k(M)\to\f^k(M)$ 
(as for $\Box$, we use the same symbol for each $k$). Their main properties are collected below:
\begin{enumerate}
\item[1.] $G_\pm$ is linear;
\item[2.] $\Box G_\pm\alpha=\alpha$ for each $\alpha\in\fc^k(M)$;
\item[3.] $G_\pm\Box\alpha=\alpha$ for each $\alpha\in\fc^k(M)$;
\item[4.] $\supp(G_\pm\alpha)\subseteq J_M^\pm\big(\supp(\alpha)\big)$ for each $\alpha\in\fc^k(M)$;
\item[5.] $(G_\mp\alpha,\beta)=(\alpha,G_\pm\beta)$ for each $\alpha,\beta\in\fc^k(M)$.
\end{enumerate}
Existence and uniqueness of $G_\pm$ fulfilling properties 1--4 follow from $\Box$ being normally hyperbolic, 
while 5 is a consequence of formal self-adjointness. 
For a thorough analysis of the properties of advanced and retarded Green operators 
for normally hyperbolic operators on globally hypoerbolic spacetimes 
the reader should refer to the literature, e.g.\ \cite[Section\ 3.4]{BGP07}. 

Further relevant properties of $\Box$ and of its advanced and retarded Green operators are shown in \cite{Pfe09}. 
For convenience, we recollect them here: 
\begin{align*}
\dd\Box\alpha & =\Box\dd\alpha\,, & \de\Box\alpha & =\Box\de\alpha\,, & \forall\alpha\in\f^k(M)\,;\\
\dd G_\pm\beta & =G_\pm\dd\beta\,, & \de G_\pm\beta & =G_\pm\de\beta\, & \forall\beta\in\fc^k(M)\,.
\end{align*}
The identities on the second line follow from those on the first one. 

We will often make use of the so-called {\em causal propagator} $G=G_+-G_-$ for $\Box$. 
Its properties follow from those of $G_\pm$ and can be condensed in the following exact sequence, 
see\ \cite[Theorem 3.4.7]{BGP07}: 
\begin{equation*}
0\lra\fc^k(M)\overset{\Box}{\lra}\fc^k(M)\overset{G}{\lra}\fsc^k(M)\overset{\Box}{\lra}\fsc^k(M)\lra0\,.
\end{equation*}
Surjectivity of $\Box:\fsc^k(M)\to\fsc^k(M)$ follows from \cite[Corollary 5]{Gin09}. 

As explained in \cite{Bar13,San13}, there are unique extensions of the advanced and retarded Green operators: 
\begin{align*}
G_+:\fpc^k(M)\to\f^k(M)\,,\quad G_-:\ffc^k(M)\to\f^k(M)\,,
\end{align*}
where the subscripts $\mathrm{pc}$ and $\mathrm{fc}$ refer to 
past compact and respectively future compact supports. 
Similarly, one has a unique extension $G:\ftc^k(M)\to\f^k(M)$ of the causal propagator. 
All properties of $G_\pm$ and $G$ carry over to their extensions.

\subsection{Generalized vector potential}
We consider a fixed globally hyperbolic spacetime $M$ of dimension $m$. 
Field configurations are $k$-forms $A$ with $k\in\{1,\dots,m-1\}$ such that $\de\dd A=0$ up to gauge equivalence. 
This means two $k$-forms $A,A^\prime$ are regarded 
as equivalent provided $A^\prime=A+\dd\chi$, $\chi\in\f^{k-1}(M)$. 
Notice that for $k=1$, we are dealing with the case of the vector potential for electromagnetism 
without external sources, see\ \cite{Dim92,FP03,Dap11} for a discussion including external sources as well. 
The model for arbitrary $k$ (including external sources as well) has been analyzed in \cite{Pfe09,SDH12}. 

\begin{remark}\label{remEquivalentVecPot}
Note that the equation $\dd\de(\ast A)=0$ is equivalent to $\de\dd A=0$. 
In fact, replacing the model described above by one 
where the field configuration $B\in\f^k(M)$ fulfils the equation of motion $\dd\de B=0$ 
and whose gauge equivalence is defined by the condition 
$B^\prime\sim B\iff B^\prime=B+\de\chi$, $\chi\in\f^{k+1}(M)$, 
provides a completely equivalent system, related to the original one via the Hodge star operator $\ast$. 
\end{remark}

As a first step, we characterize the space of field configurations up to gauge. 
That done, a suitable space of classical observables will be introduced. 
In the end, Theorem\ \ref{thmSCTCPoincareDuality} will provide a sound motivation 
for the choice of this space of observables. 

We introduce the space of solutions to the field equation $\de\dd A=0$: 
\begin{equation*}
\S_A=\ker\big(\de\dd:\f^k(M)\to\f^k(M)\big)\,.
\end{equation*}
As mentioned above, two solutions are regarded as equivalent 
whenever they differ by $\dd\chi$ for some $\chi\in\f^{k-1}(M)$, i.e.\ their difference lies in 
\begin{equation*}
\G_A=\dd\f^{k-1}(M)\,.
\end{equation*}
Therefore, the space of gauge classes of on-shell field configurations is the quotient $\S_A/\G_A$. 
A convenient characterization of this space is provided below. 

\begin{lemma}\label{lemLorenzGaugeFixing}
Let $M$ be a globally hyperbolic spacetime. 
Each $A\in\S_A$ is gauge equivalent to a configuration $A^\prime\in\S_A$ satisfying the Lorenz gauge, 
namely there exists $\chi\in\f^{k-1}(M)$ such that $A^\prime=A+\dd\chi$ fulfils $\de A^\prime=0$. 
\end{lemma}

\begin{proof}
Take $A\in\S_A$ and consider the equation $\de\dd\chi=-\de A$. 
Consider a partition of unit $\{f_+,f_-\}$ on $M$ such that 
$f_+=1$ in a past compact region, while $f_-=1$ in a future compact one. 
Denoting with $G_\pm$ the advanced/retarded Green operator for $\Box$, 
one can explicitly write down a solution of $\de\dd\chi=-\de A$, 
namely $\chi=-\de\big(G_+(f_+A)+G_-(f_-A)\big)$. 
Setting $A^\prime=A+\dd\chi$ one realizes that $\de A^\prime=\de A+\de\dd\chi=0$ and $A^\prime\in\S_A$. 
\end{proof}

\begin{lemma}\label{lemLorenzSol}
Let $M$ be a globally hyperbolic spacetime and denote with $G$ the causal propagator for $\Box$. 
Each $A\in\S_A$ in the Lorenz gauge is gauge equivalent to $G\omega\in\S_A$ with $\omega\in\ftcde^{k}(M)$, 
namely there exist $\chi\in\f^{k-1}(M)$ and $\omega\in\ftcde^k(M)$ such that $G\omega=A+\dd\chi$. 
\end{lemma}

\begin{proof}
Take $A\in\S_A$ such that $\de A=0$. Then $\Box A=0$, hence $A=G\theta$ for a suitable $\theta\in\ftc^k(M)$. 
The Lorenz gauge condition ensures the existence of $\rho\in\ftc^{k-1}(M)$ such that $\Box\rho=\de\theta$. 
Introducing $\omega=\theta-\dd\rho\in\ftc^k(M)$ and $\chi=-G\rho\in\f^{k-1}(M)$ 
and noting that $\de\rho=0$, one reads $\de\omega=0$ and $G\omega=A+\dd\chi$. 
\end{proof}

\begin{theorem}\label{thmOnShellA}
Let $M$ be a globally hyperbolic spacetime and denote with $G$ the causal propagator for $\Box$. 
The following is an isomorphism of vector spaces: 
\begin{equation*}
\ftcde^k(M)/\de\dd\ftc^k(M)\to\S_A/\G_A\,,\quad[\omega]\mapsto[G\omega]\,.
\end{equation*}
\end{theorem}

\begin{proof}
For $\omega\in\ftcde^k(M)$, $G\omega$ is coclosed, therefore $\de\dd G\omega=\Box G\omega=0$. 
This means that $G$ maps $\ftcde^k(M)$ to $\S_A$, hence we can consider the map 
\begin{equation*}
\ftcde^k(M)\to\S_A/\G_A\,,\quad\omega\mapsto[G\omega]\,.
\end{equation*} 

We prove surjectivity of this map: Given $[A]\in\S_A/\G_A$, fixing a representative $A\in[A]$ 
and exploiting Lemma\ \ref{lemLorenzGaugeFixing} in the first place and then Lemma\ \ref{lemLorenzSol}, 
we find $\omega\in\ftcde^{k}(M)$ such that $G\omega$ is gauge equivalent to $A$, i.e\ $[G\omega]=[A]$. 

It remains only to check that the kernel of the map mentioned above coincides with $\de\dd\ftc^k(M)$. 
The inclusion in one direction follows from the identity 
$G\de\dd\theta=G(\Box-\dd\de)\theta=-\dd G\de\theta\in\G_A$ for each $\theta\in\ftc^k(M)$, 
hence $\de\dd\ftc^k(M)$ is included in the kernel. 
For the converse inclusion, consider $\omega\in\ftcde^k(M)$ 
such that $G\omega=\dd\chi$ for a suitable $\chi\in\f^{k-1}(M)$. This entails $\de\dd\chi=0$. 
For $k\geq2$, via Lemma\ \ref{lemLorenzGaugeFixing} and Lemma\ \ref{lemLorenzSol}, 
one finds $\rho\in\ftcde^{k-1}(M)$ and $\xi\in\f^{k-2}(M)$ such that $G\rho=\chi+\dd\xi$. 
For $k=1$,  $\Box\chi=\de\dd\chi=0$ and hence one finds $\rho\in\ctc(M)$ such that $G\rho=\chi$. 
Therefore $G\omega=G\dd\rho$ and $\de\dd\rho=\Box\rho$ in both cases, 
hence there exists $\theta\in\ftc^k(M)$ such that $\omega=\dd\rho+\Box\theta$. 
Applying $\de$ to both sides, we deduce $\Box(\rho+\de\theta)=0$. 
This entails $\rho+\de\theta=0$, from which we conclude $\omega=\de\dd\theta$.  
\end{proof}

\begin{remark}\label{remSCOnShellA}
Notice that one could consider the space $\SscA$ of solutions 
with spacelike compact support for the field equation $\de\dd A=0$. 
Accordingly, one has to consider a stricter notion of gauge equivalence provided by $\GscA=\dd\fsc^{k-1}(M)$. 
Using exactly the same arguments, one can prove statements similar to 
Lemma\ \ref{lemLorenzGaugeFixing}, Lemma\ \ref{lemLorenzSol} and Theorem\ \ref{thmOnShellA}, 
where $\f^k$ and $\ftc^k$ are replaced respectively by $\fsc^k$ and $\fc^k$. 
\end{remark}

We focus now the attention on the construction of suitable linear observables 
for the space of field configurations $\S_A/\G_A$. We follow the spirit of \cite{BFR12}, 
hence classical observables are regarded as functionals on the space of field configurations. 
Since we are dealing with linear equations of motion, our approach is very close 
to strategy followed in \cite{BDS13a} for affine field theories 
and its modification developed in \cite{BDS13b} to include the case of gauge theories. 

As a first step, for each $\alpha\in\fc^k(M)$, we consider the linear functional 
\begin{equation*}
\O_\alpha:\f^k(M)\to\bbR\,,\quad\O_\alpha(\beta)=(\alpha,\beta)\,,
\end{equation*}
where $(\cdot,\cdot)$ refers to the pairing \eqref{eqPairing2}. 
Doing so, we define the space of off-shell linear functionals according to 
\begin{equation*}
\Ekin_A=\{\O_\alpha\,:\;\alpha\in\fc^k(M)\}\simeq\fc^k(M)\,,
\end{equation*}
where the isomorphism is provided by non-degeneracy of $(\cdot,\cdot)$. Bearing in mind this isomorphism, 
sometimes we might refer to $\fc^k(M)$ as the space of off-shell linear functionals. 

Not all functionals $\O_\alpha$, $\alpha\in\fc^k(M)$, are invariant under gauge transformations. 
Since we regard configurations differing by a gauge transformation as being equivalent, 
linear functionals which are meant to define observables should be gauge invariant. 
Therefore we restrict to a subspace of linear functionals $\Einv_A\subseteq\Ekin_A$ 
characterized by the property $\O_\alpha\in\Einv_A\iff\O_\alpha(\G_A)=\{0\}$. 
Theorem\ \ref{thmStokes}, together with non-degeneracy of the pairing $(\cdot,\cdot)$, 
provides a convenient characterization for the space of gauge invariant linear functionals: 
\begin{equation*}
\Einv_A=\{\O_\alpha\,:\;\alpha\in\fcde^k(M)\}\simeq\fcde^k(M)\,.
\end{equation*}
These are all the functionals of $\Ekin_A$ which can be consistently evaluated on $\f^k(M)/\G_A$, 
the space of gauge classes of off-shell field configurations. 
By the same argument as above and the inclusion $\fcde^k(M)\subseteq\fc^k(M)$, 
$\Einv_A$ is identified with the isomorphic space $\fcde^k(M)$. 

Up to this point no information about the dynamics of the system has been taken into account. 
To encode dynamics on the space of gauge invariant linear functionals, we force them 
not to be defined for off-shell configurations $A$, namely such that $\de\dd A\neq0$. 
This is obtained taking the quotient of $\Einv_A$ 
by the image of the formal adjoint of the equation of motion operator $\de\dd$. 
Since $(\de\dd\alpha,\beta)=(\alpha,\de\dd\beta)$ for each $\alpha,\beta\in\f^k(M)$ 
such that the intersection of their supports is compact 
(see\ Theorem\ \ref{thmStokes} and the following discussion), $\de\dd$ is the formal adjoint of itself. 
Therefore we define the space of linear observables for $\S_A/\G_A$ as the quotient 
\begin{equation}\label{eqObsA}
\E_A=\Einv_A/\de\dd\fc^k(M)\,.
\end{equation}
The evaluation of $[\alpha]\in\E_A$ on $[A]\in\S_A/\G_A$ is defined by the evaluation on arbitrary representatives, 
namely we set $\O_{[\alpha]}\big([A]\big)=\O_\alpha(A)$ 
for $\alpha\in[\alpha]$ and $A\in[A]$ chosen arbitrarily. 
This definition is well-posed since the difference between 
two representative $\alpha,\alpha^\prime\in[\alpha]$ is an element of $\de\dd\fc^k(M)$, 
while the difference between two representatives $A,A^\prime\in[A]$ lies in $\G_A$. 

The aim is to show that, provided $M$ admits a finite good cover, $\E_A$ is the correct space of linear observables 
for the space of gauge classes of on-shell field configurations $\S_A/\G_A$. 
This is to be intended in the following sense: 
\begin{enumerate}
\item[a.] $\E_A$ contains sufficiently many functionals to distinguish between different elements in $\S_A/\G_A$, 
hence classically there is no reason to take into account more general functionals; 
\item[b.] There are no redundant elements in $\E_A$. 
As a metter of fact, given two elements $[\alpha]\neq[\beta]\in\E_A$, 
there exists an on-shell configuration $[A]\in\S_A/\G_A$ such that the evaluation of $[\alpha]$ and $[\beta]$ 
gives different results, namely $\O_{[\alpha]}\big([A]\big)\neq\O_{[\beta]}\big([A]\big)$. 
This fact entails that any further quotient of $\E_A$ would lead to functionals 
which are not well-defined on the whole space of gauge classes of on-shell field configurations $\S_A/\G_A$. 
Therefore, at least from a classical viewpoint, one should not interpret 
any quotient of $\E_A$ as consisting of observables for $\S_A/\G_A$ 
since there would exist configurations which cannot be tested in a consistent way. 
\end{enumerate}
The first result is achieved using standard Poincar\'e duality, Theorem\ \ref{thmPoincareDuality}, 
while the second follows exploiting Poincar\'e duality between spacelike compact 
and timelike compact cohomologies, Theorem\ \ref{thmSCTCPoincareDuality}. 
These facts are recollected in the theorem stated below. 

\begin{theorem}\label{thmObsA}
On a globally hyperbolic spacetime $M$ the following holds: 
\begin{enumerate}
\item[a.] If $[A]\in\S_A/\G_A$ is such that 
$\O_{[\alpha]}\big([A]\big)=0$ for each $[\alpha]\in\E_A$, then $[A]=0$; 
\item[b.] If $M$ admits a finite good cover and $[\alpha]\in\E_A$ is such that 
$\O_{[\alpha]}\big([A]\big)=0$ for each $[A]\in\S_A/\G_A$, then $[\alpha]=0$. 
\end{enumerate}
\end{theorem}

\begin{proof}
Let us start from the first statement. Fixing an arbitrary representative $A\in[A]$, 
the hypothesis means that $(\alpha,A)=0$ for each $\alpha\in\fcde^k(M)$. 
Restricting to $\alpha$ of the form $\de\beta$, $\beta\in\fc^{k+1}(M)$, one deduces $\dd A=0$, 
therefore $[A]$ can be regarded as an element of $\hdd^k(M)$. 
Taking into account the pairing ${}_\de(\cdot,\cdot)$ of \eqref{eqPoincarePairing}, 
the hypothesis translates into ${}_\de([\alpha],[A])=0$ for each $[\alpha]\in\hcde^k(M)$, 
meaning that ${}_\de(\cdot,[A])$ is the trivial element of $\big(\hcde^k(M)\big)^\ast$.  
Remark \ref{remPoincareDualityImproved} entails that 
the map $[A]\in\hdd^k(M)\mapsto{}_\de(\cdot,[A])\in\big(\hcde^k(M)\big)^\ast$ is an isomorphism, 
therefore $[A]$ is the trivial cohomology class in $\hdd^k(M)$, 
meaning that $[A]=0$ in the sense of $\S_A/\G_A$ as well. 

To prove the second statement, we exploit Theorem\ \ref{thmOnShellA} in order to rephrase the hypothesis as 
$\O_{[\alpha]}\big([G\omega]\big)=0$ for each $\omega\in\ftcde^k(M)$. 
Fixing a representative $\alpha\in[\alpha]$ and exploiting the properties of $G$, 
the hypothesis reads $(G\alpha,\omega)=0$ for each $\omega\in\ftcde^k(M)$. 
Taking into account $\omega$ of the form $\de\xi$, $\xi\in\fc^{k+1}(M)$, one deduces $\dd G\alpha=0$. 
Since $\supp(G\alpha)$ is spacelike compact, $G\alpha$ can be regarded as a representative 
of the spacelike compact cohomology class $[G\alpha]\in\hscdd^k(M)$. 
From the hypothesis, $([G\alpha],[\omega])_\de=0$ for each $[\omega]\in\htcde^k(M)$, 
where $(\cdot,\cdot)_\de$ denotes the pairing in \eqref{eqSCTCPairing}. 
Since $M$ has a finite good cover, Theorem\ \ref{thmSCTCPoincareDuality} entails $[G\alpha]=0$, 
hence there exists $\chi\in\fsc^{k-1}(M)$ such that $\dd\chi=G\alpha$. 
From this fact, together with $\de\alpha=0$, the identity $\de\dd\chi=0$ follows. 
For $k\geq2$, Lemma\ \ref{lemLorenzGaugeFixing} and Lemma\ \ref{lemLorenzSol}, 
together with Remark \ref{remSCOnShellA}, provide $\tilde{\chi}\in\fcde^{k-1}(M)$ 
and $\hat{\chi}\in\fsc^{k-2}(M)$ such that $G\tilde{\chi}=\chi+\dd\hat{\chi}$. 
For $k=1$, $\Box\chi=\de\dd\chi=0$, therefore one finds $\tilde{\chi}\in\cc(M)$ such that $G\tilde{\chi}=\chi$. 
In both cases one deduces $G\dd\tilde{\chi}=G\alpha$ and $\de\dd\tilde{\chi}=\Box\tilde{\chi}$. 
The first identity implies $\dd\tilde{\chi}=\alpha+\Box\beta$ for a suitable $\beta\in\fc^k(M)$, 
while, applying $\de$ to both sides and taking into account the second identity too, 
one shows that $\tilde{\chi}=\de\beta$. We conclude that $[\alpha]=[-\de\dd\beta]=0$ in $\E_A$. 
\end{proof}

\begin{remark}
For $k=1$ the second statement of Theorem\ \ref{thmObsA} has interesting implications 
for the spaces of classical observables considered in \cite{BDS13b,BDHS13}. 

In the case of \cite[Section\ 3]{BDS13b} observables are affine functionals whose linear part lies in $\E_A$, 
see\ \eqref{eqObsA}. Furthermore, the linear part of the on-shell condition 
for \cite{BDS13b} coincides with the equation of motion for $A$. 
Therefore, on account of the second statement of Theorem\ \ref{thmObsA}, 
on the space of classical observables of \cite{BDS13b} it is not possible to take any further quotient 
affecting the linear part without restricting at the same time the space of configurations 
on which the equivalence classes of the resulting space of functionals are supposed to be evaluated. 
In fact, in \cite[Section\ 7]{BDS13b} the authors take a quotient of the space of observables 
and simultaneously they restrict the on-shell condition for field 
configurations.\footnote{It is required that on-shell field configurations have zero charge, 
namely that their curvature is a coexact $2$-form, hence, according to Gauss law, they carry null charge.} 

In the case of \cite{BDHS13} it is shown that observables are sufficiently many 
to distinguish between gauge classes of on-shell field configurations, see\ \cite[Theorem 3.2]{BDHS13}. 
Theorem\ \ref{thmObsA} above entails also that there are no redundant observables in \cite{BDHS13}, 
meaning that one can always find a configuration such that 
two different observables provide a different outcome upon evaluation. 
This result is a consequence of two facts: 
First, the linear part of the exponent of an observable in \cite{BDHS13} lies in $\E_A$, 
see\ \cite[eq.\ (3.5)]{BDHS13}. 
Second, once fixed an on-shell configuration as a reference, all others are obtained adding $A\in\S_A$. 
Therefore, for the linear part of the exponent we can apply Theorem\ \ref{thmObsA}. 
That done, it is easy to check that also the purely affine parts coincide. 
This shows that the space of classical observables introduced in \cite[Section\ 3]{BDHS13} 
for the field theoretical model considered there is the correct one with respect to the criteria presented in this paper. 
\end{remark}

\subsection{Generalized Faraday tensor}
The background is provided by a $m$-dimensional globally hyperbolic spacetime $M$. 
In the spirit of electromagnetism, we consider a $k$-form $F$, $k\in\{1\,\dots,m-1\}$. 
The dynamics for $F$ is introduced imposing $\dd F=0$ and $\de F=0$. 
Notice that for $k=2$, $F$ can be interpreted as the Faraday tensor of electromagnetism (without sources), 
see\ e.g.\ \cite{Bon77,DL12} for the case where also an external current is considered. 

More formally, one can regard $\f^k(M)$ as the space of off-shell field configurations. 
Introducing also the linear differential operator 
\begin{equation}\label{eqddboxplusde}
\dd\boxplus\de:\f^k(M)\to\f^{k+1}(M)\oplus\f^{k-1}(M)\,,
\quad(\dd\boxplus\de)\alpha=\dd\alpha\oplus\de\alpha\,, 
\end{equation}
which rules the dynamics of the field, one can specify the on-shell condition according to $(\dd\boxplus\de)F=0$. 
Therefore the space of on-shell field configurations is given by 
\begin{equation*}
\S_F=\ker\big(\dd\boxplus\de:\f^k(M)\to\f^{k+1}(M)\oplus\f^{k-1}(M)\big)\,.
\end{equation*}

The following theorem provides a convenient characterization of $\S_F$.

\begin{theorem}\label{thmOnShellF}
Let $M$ be a globally hyperbolic spacetime and denote with $G$ the causal propagator for $\Box$. 
Then the following is an isomorphism of vector spaces: 
\begin{equation*}
\frac{\ftcdd^{k+1}(M)\oplus\ftcde^{k-1}(M)}{(\dd\boxplus\de)\ftc^k(M)}\to\S_F\,,
\quad[\alpha\oplus\beta]\mapsto G(\de\alpha+\dd\beta)\,.
\end{equation*}
\end{theorem}

\begin{proof}
Consider the linear differential operator 
\begin{equation}\label{eqdeplusdd}
\de+\dd:\f^{k+1}(M)\oplus\f^{k-1}(M)\to\f^k(M)\,,\quad(\de+\dd)(\alpha\oplus\beta)=\de\alpha+\dd\beta\,.
\end{equation}
Note that $(\dd\boxplus\de)(\de+\dd)(\omega\oplus\theta)=\Box\omega\oplus\Box\theta$ 
for $\omega\in\fdd^{k+1}(M)$ and $\theta\in\fde^{k-1}(M)$. 
Taking $\omega=G\alpha$ and $\theta=G\beta$ for $\alpha\in\ftcdd^{k+1}(M)$ and $\beta\in\ftcde^{k-1}(M)$, 
one reads 
\begin{equation*}
(\dd\boxplus\de)G(\de\alpha+\dd\beta)=(\dd\boxplus\de)(\de+\dd)(G\alpha\oplus G\beta)
=\Box G\alpha\oplus\Box G\beta=0
\end{equation*}
on account of the identities $\dd G=G\dd$ and $\de G=G\de$ on $\ftc^k(M)$. 
Therefore $G(\de+\dd)$ provides a linear map from $\ftcdd^{k+1}(M)\oplus\ftcde^{k-1}(M)$ to $\S_F$. 

We check surjectivity of this map. Any $F\in\S_F$ is both closed and coclosed, therefore $\Box F=0$. 
The properties of the causal propagator entail that we can express $F$ as $G\omega$ 
for a suitable $\omega\in\ftc^k(M)$. $\dd F=0$ and $\de F=0$ ensure 
the existence of $\alpha\in\ftc^{k+1}(M)$ and $\beta\in\ftc^{k-1}(M)$ 
such that $\Box\alpha=\dd\omega$ and $\Box\beta=\de\omega$, 
thus implying $\dd\alpha=0$ and $\de\beta=0$ as well as $\omega=\de\alpha+\dd\beta$. 

We are left with the characterization of the kernel of the map under consideration. 
It is clear that $(\de+\dd)(\dd\boxplus\de)=\Box$, hence $(\dd\boxplus\de)\ftc^k(M)$ is a subspace of the kernel. 
For the converse inclusion, consider $\alpha\oplus\beta\in\ftcdd^{k+1}(M)\oplus\ftcde^{k-1}(M)$ 
such that $G(\de\alpha+\dd\beta)=0$. 
It follows that $\de\alpha+\dd\beta=\Box\omega$ for a suitable $\omega\in\ftc^k(M)$. 
As a consequence, $\alpha=\dd\omega$ and $\beta=\de\omega$, thus concluding the proof. 
\end{proof}

\begin{remark}\label{remSCOnShellF}
As in Remark \ref{remSCOnShellA}, one has a way to represent the space of spacelike compact solutions $\SscF$
of the equations of motion $\dd F=0$ and $\de F=0$ similar to the one of Theorem\ \ref{thmOnShellF}: 
\begin{equation*}
\frac{\fcdd^{k+1}(M)\oplus\fcde^{k-1}(M)}{(\dd\boxplus\de)\fc^k(M)}\to\SscF\,,
\quad[\alpha\oplus\beta]\mapsto G(\de\alpha+\dd\beta)\,.
\end{equation*}
The argument to prove that the map above is an isomorphism is identical to the proof of 
Theorem\ \ref{thmOnShellF} except for the fact that forms are to be replaced by spacelike compact ones, 
while timelike compact forms are to be replaced by compactly supported ones. 
\end{remark}

Observables for the classical field theory of $F$ are defined as functionals on field configurations 
in the spirit of \cite{BFR12}. However, the equation of motion being linear, 
we are allowed to restrict to linear functionals, therefore our approach mimics that of \cite{BDS13a}. 
First a space of off-shell linear functionals in terms of $k$-forms with compact support is introduced: 
\begin{equation*}
\Ekin_F=\{\O_\alpha\,:\;\alpha\in\fc^k(M)\}\simeq\fc^k(M)\,,
\end{equation*}
$\O_\alpha:\f^k(M)\to\bbR$ being defined by the formula $\O_\alpha(\beta)=(\alpha,\beta)$ 
for $\beta\in\f^k(M)$, where $(\cdot,\cdot)$ is defined in \eqref{eqPairing2}. 
From non-degeneracy of $(\cdot,\cdot)$ it follows that $\Ekin_F$ is isomorphic to $\fc^k(M)$. 
This isomorphism will be considered as an identification in the following. 
In particular, sometimes we will refer to $\fc^k(M)$ as the space of off-shell linear functionals. 
At this stage, $\Ekin_F$ cannot be interpreted as a space of observables for the field theory of $F$ 
since no information about the dynamics of $F$ is encoded. Specifically, given $\alpha\in\fc^k(M)$, 
$\O_\alpha$ can be evaluated on any off-shell configuration $\beta\in\f^k(M)$, 
while a proper observable should be defined only on $\S_F$, namely for on-shell configurations. 

To encode the dynamics for $F$ on $\Ekin_F$, we proceed dually. 
Introducing the pairing on $\f^{k+1}(M)\oplus\f^{k-1}(M)$ 
induced by the pairing $(\cdot,\cdot)$ on each summand, see\ \eqref{eqPairing2}, 
one realizes that $\de+\dd:\f^{k+1}(M)\oplus\f^{k-1}(M)\to\f^k(M)$, eq.\ \eqref{eqdeplusdd}, 
is the formal adjoint of $\dd\boxplus\de:\f^k(M)\to\f^{k+1}(M)\oplus\f^{k-1}(M)$, eq.\ \eqref{eqddboxplusde}. 
This fact is a direct consequence of the identity 
$(\de\alpha+\dd\beta,\gamma)=(\alpha,\dd\gamma)+(\beta,\de\gamma)$, 
which holds true for each $\alpha\in\f^{k+1}(M)$, $\beta\in\f^{k-1}(M)$ and $\gamma\in\f^k(M)$ 
such that both $\supp(\alpha)\cap\supp(\gamma)$ and $\supp(\beta)\cap\supp(\gamma)$ are compact. 
This result follows from $\delta$ being the formal adjoint of $\dd$, 
see the discussion following Theorem\ \ref{thmStokes}. 
This suggests to define the space of linear observables for $\S_F$ as the quotient below: 
\begin{equation*}
\E_F=\Ekin_F/(\de+\dd)\big(\fc^{k+1}(M)\oplus\fc^{k-1}(M)\big)\,.
\end{equation*}
An element $[\alpha]\in\E_F$ can be interpreted as an observable for on-shell field configurations 
since $\O_\alpha(F)$ does not depend on the choice of a representative $\alpha\in[\alpha]$ 
provided $F$ is on-shell, namely $F\in\S_F$. This implicitly defines linear functionals 
$\O_{[\alpha]}:\S_F\to\bbR$ according to $\O_{[\alpha]}(F)=\O_\alpha(F)$ 
for each $F\in\S_F$ and regardless of the choice of $\alpha\in[\alpha]$. 

Under the assumption that $M$ admits a finite good cover, 
the following theorem enforces the interpretation of $\E_F$ as the correct space of classical linear observables 
for the space of on-shell field configurations $\S_F$ in the sense stated below: 
\begin{enumerate}
\item[a.] $\E_F$ has sufficiently many elements to distinguish between different on-shell configurations. 
In this respect, there is no reason to consider a larger space of classical observables; 
\item[b.] Given two different elements $[\alpha],[\beta]\in\E_F$, 
there exists an on-shell configuration $F\in\S_F$ such that the evaluation of $[\alpha]$ and $[\beta]$ 
gives different results, namely $\O_{[\alpha]}(F)\neq\O_{[\beta]}(F)$. 
This fact entails that no further quotient of $\E_F$ is allowed 
if one wants to consistently test elements of this quotient on the whole $\S_F$. 
Accordingly, from a classical perspective, one should not interpret any quotient of $\E_F$ 
as consisting of observables for the space of on-shell configurations $\S_F$. 
As a matter of fact there would exist on-shell configurations 
that cannot be tested in a consistent way by the elements of any further quotient of $\E_F$. 
\end{enumerate}
The first statement follows from non-degeneracy of the pairing $(\cdot,\cdot)$ between $k$-forms 
defined in \eqref{eqPairing2}, while the second is a consequence of Poincar\'e duality between cohomologies 
with spacelike compact and timelike compact support, Theorem\ \ref{thmSCTCPoincareDuality}. 

\begin{theorem}
On a globally hyperbolic spacetime $M$ the following holds: 
\begin{enumerate}
\item[a.] If $F\in\S_F$ is such that $\O_{[\alpha]}(F)=0$ for each $[\alpha]\in\E_F$, then $F=0$; 
\item[b.] If $M$ admits a finite good cover and $[\alpha]\in\E_F$ is such that 
$\O_{[\alpha]}(F)=0$ for each $F\in\S_F$, then $[\alpha]=0$. 
\end{enumerate}
\end{theorem}

\begin{proof}
The first statement follows from non-degeneracy of the pairing $(\cdot,\cdot):\fc^k(M)\times\f^k(M)\to\bbR$ 
defined in eq.\ \eqref{eqPairing2}. 

Let us consider $[\alpha]$ as in the second statement. 
For each $F\in\S_F$, Theorem\ \ref{thmOnShellF} entails the existence 
of $\omega\in\ftcdd^{k+1}(M)$ and $\theta\in\ftcde^{k-1}(M)$ such that $F=G(\de\omega+\dd\theta)$. 
Exploiting the properties of $\dd$, $\de$ and $G$, the hypothesis can be rephrased as 
$(\dd G\alpha,\omega)+(\de G\alpha,\theta)=0$ for each $\omega\in\ftcdd^{k+1}(M)$ 
and each $\theta\in\ftcde^{k-1}(M)$, therefore both contributions must vanish identically. 
Choosing $\omega$ and $\theta$ respectively of the form $\dd\rho$ and $\de\sigma$, $\rho,\sigma\in\fc^k(M)$, 
one shows that both $\de\dd G\alpha$ and $\dd\de G\alpha$ vanish. 
Since $G\alpha$ has spacelike compact support and recalling \eqref{eqSCTCPairing}, 
we deduce ${}_\de([\dd G\alpha],[\omega])=0$ for each $[\omega]\in\htcdd^{k+1}(M)$ 
and $([\de G\alpha],[\theta])_\de$ for each $[\theta]\in\htcde^{k-1}(M)$. 
On account of Theorem\ \ref{thmSCTCPoincareDuality}, $[\dd G\alpha]$ and $[\de G\alpha]$ 
are the trivial elements of $\hscde^{k+1}(M)$ and respectively of $\hscdd^{k-1}(M)$. 

At this point we distinguish three possibilities: 
\begin{description}
\item[$\boldsymbol{2\leq k\leq m-2}$] Since $[\dd G\alpha]=0$ in $\hscde^{k+1}(M)$ 
and $[\de G\alpha]=0$ in $\hscdd^{k-1}(M)$, there exist $\xi\in\fsc^{k+2}(M)$ and $\eta\in\fsc^{k-2}(M)$ 
such that $\de\xi=\dd G\alpha$ and $\dd\eta=\de G\alpha$. 
Applying $\dd$ to both sides in the first case and $\de$ in the second, we get $\dd\de\xi=0$ and $\de\dd\eta=0$. 
On account of the discussion for the vector potential, 
see\ Remark\ \ref{remEquivalentVecPot}, Remark \ref{remSCOnShellA}, 
Lemma\ \ref{lemLorenzGaugeFixing} and Lemma\ \ref{lemLorenzSol}, 
for $3\leq k\leq m-3$ one finds $\tilde{\xi}\in\fcdd^{k+2}(M)$, $\hat{\xi}\in\fsc^{k+3}(M)$, 
$\tilde{\eta}\in\fcde^{k-2}(M)$ and $\hat{\eta}\in\fsc^{k-3}(M)$ 
such that $G\tilde{\xi}=\xi+\de\hat{\xi}$ and $G\tilde{\eta}=\eta+\dd\hat{\eta}$. 
For $k=2$ and $m\geq5$, nothing changes on the $\xi$-side, while $\de\dd\eta=0$ reduces to $\Box\eta=0$, 
therefore one finds $\tilde{\eta}\in\cc(M)$ such that $G\tilde{\eta}=\eta$. 
Similarly, for $k=m-2$ and $m\geq5$, one reads $\Box\xi=0$, ensuring the existence of $\tilde\xi\in\fc^m(M)$ 
such that $G\tilde{\xi}=\xi$, while nothing changes on the $\eta$-side. 
For $k=2$ and $m=4$, $\dd\de\xi=0$ and $\de\dd\eta=0$ reduce to $\Box\xi=0$ and $\Box\eta=0$, 
hence $\xi=G\tilde{\xi}$ and $\eta=G\tilde{\eta}$ for suitable $\tilde{\xi}\in\fc^4(M)$ and $\eta\in\cc(M)$. 
In any case we obtain $G\de\tilde{\xi}=G\dd\alpha$ and $G\dd\tilde{\eta}=G\de\alpha$, 
from which the identities $\de\tilde{\xi}=\dd\alpha-\Box\mu$ and $\dd\tilde{\eta}=\de\alpha-\Box\nu$ 
for suitable $\mu\in\fc^{k+1}(M)$ and $\nu\in\fc^{k-1}(M)$ follow. 
In addition to that, one has $\dd\de\tilde{\xi}=\Box\tilde{\xi}$ and $\de\dd\tilde{\eta}=\Box\tilde{\eta}$, 
therefore $\tilde{\xi}=-\dd\mu$ and $\tilde{\eta}=-\de\nu$, 
hence $\dd\de\mu=\dd\alpha$ and $\de\dd\nu=\de\alpha$. 
\item[$\boldsymbol{k=1}$ and $\boldsymbol{m\geq3}$] The calculation for $\dd G\alpha$ is left unchanged, 
while $[\de G\alpha]=0$ in $\hscdd^0(M)$ means that $\de G\alpha=0$,  
therefore there exists $\nu\in\cc(M)$ such that $\de\alpha=\Box\nu=\de\dd\nu$. 
\item[$\boldsymbol{k=m-1}$ and $\boldsymbol{m\geq3}$] Nothing changes for $\de G\alpha$ 
with respect to the case $2\leq k\leq m-2$, while $[\dd G\alpha]=0$ in $\hscde^m(M)$ means that 
$\dd G\alpha=0$. Therefore $\dd\alpha=\Box\mu=\dd\de\mu$ for a suitable $\mu\in\fc^m(M)$. 
\item[$\boldsymbol{k=1}$ and $\boldsymbol{m=2}$] In this case the fact that $[\dd G\alpha]$ 
and $[\de G\alpha]$ are the trivial elements of $\hscde^{2}(M)$ and respectively of $\hscdd^{0}(M)$ 
can be rephrased as $\dd G\alpha=0$ and $\de G\alpha=0$, therefore one can find $\mu\in\fc^2(M)$ and 
$\nu\in\cc(M)$ such that $\dd\alpha=\Box\mu=\dd\de\mu$ and $\de\alpha=\Box\nu=\de\dd\nu$. 
\end{description}
Summing up, for any $k\in\{1,\dots,m-1\}$ one has $\dd\de\mu=\dd\alpha$ and $\de\dd\nu=\de\alpha$ 
for suitable $\mu\in\fc^{k+1}(M)$ and $\nu\in\fc^{k-1}(M)$. 
Applying $\de$ on both sides of the first identity and $\dd$ on both sides of the second, 
one has $\alpha=\de\mu+\dd\nu$. 
\end{proof}

\section*{Acknowledgements}
The author wishes to thank Claudio Dappiaggi and Alexander Schenkel for very helpful discussions and suggestions. 
Furthermore, the author is grateful to Igor Khavkine for sharing his ideas 
about de Rham cohomologies with restricted support, which will appear in an upcoming work \cite{Kha14}. 
This research was supported by a Ph.D. scholarship of the University of Pavia.

\end{document}